\documentclass[british,hyphens]{article}
\PassOptionsToPackage{hyphens}{url}
\usepackage[T1]{fontenc}
\usepackage[utf8x]{inputenc}
\usepackage{color}
\usepackage{array}
\usepackage{float}
\usepackage{calc}
\usepackage{textcomp}
\usepackage{url}
\usepackage{amsmath}
\usepackage{amsthm}
\usepackage{amssymb}
\usepackage{graphicx}
\usepackage{microtype}

\makeatletter

\providecommand{\tabularnewline}{\\}

\numberwithin{equation}{section}
\numberwithin{figure}{section}
\theoremstyle{plain}
\newtheorem{thm}{\protect\theoremname}
\theoremstyle{definition}
\newtheorem{defn}[thm]{\protect\definitionname}
\theoremstyle{plain}
\newtheorem{lem}[thm]{\protect\lemmaname}
\theoremstyle{remark}
\newtheorem{rem}[thm]{\protect\remarkname}
\theoremstyle{plain}
\newtheorem{cor}[thm]{\protect\corollaryname}

\usepackage{hyperref}
\hypersetup{colorlinks = false, allcolors=., citebordercolor = [rgb]{0,0,0}, linkbordercolor = [rgb]{0,0,0}, urlbordercolor = [rgb]{0,0,0}}
\usepackage{pifont}
\usepackage[T1]{fontenc}
\usepackage{lipsum}
\usepackage{blindtext,titlefoot}
\usepackage{ragged2e}
\usepackage{float}
\usepackage{mathbbol}

\usepackage{pgfplots}
\pgfplotsset{width=7cm,compat=1.10}

\newtheorem*{assumption*}{\assumptionnumber}
\providecommand{\assumptionnumber}{}


\makeatother

\usepackage{babel}
\providecommand{\corollaryname}{Corollary}
\providecommand{\definitionname}{Definition}
\providecommand{\lemmaname}{Lemma}
\providecommand{\remarkname}{Remark}
\providecommand{\theoremname}{Theorem}

\begin{document}
\title{{\Large{}Zero-Knowledge Optimal Monetary Policy }\\
{\Large{}under Stochastic Dominance}}
\author{David Cerezo Sánchez\textsuperscript{}\\
{\small{}david@calctopia.com}}
\maketitle
\begin{abstract}
Optimal simple rules for the monetary policy of the first stochastically
dominant crypto-currency are derived in a Dynamic Stochastic General
Equilibrium (DSGE) model, in order to provide optimal responses to
changes in inflation, output, and other sources of uncertainty.\\

The optimal monetary policy \footnote{\textbf{~STATEMENT ON MONETARY POLICY GOALS AND STRATEGY}:\\
\\
The primary mandate is (stochastic) dominance.\\
The primary means of adjusting the policy stance is through changes
in money growth.\\
The monetary policy is implemented with pre-committed policy rules,
only to be revised in case of technology shocks or in the event of
a financial crisis: the stance of monetary policy will adjust as appropriate
if risks emerge that could impede the attainment of its goals, and
this document will be reviewed and updated with any changes.\\
Unlike other crypto-currencies, this monetary policy synchronises
with macro-economic observables, other fiat currencies and CBDCs:
its primary goal is to follow cooperative equilibria, falling back
to non-cooperative equilibria as last resort.} stochastically dominates all the previous crypto-currencies, thus
the efficient portfolio is to go long on the stochastically dominant
crypto-currency: a strategy-proof arbitrage featuring a higher Omega
ratio with higher expected returns, inducing an investment-efficient
Nash equilibrium over the crypto-market.\\

Zero-knowledge proofs of the monetary policy are committed on the
blockchain: an implementation is provided.\\

\textbf{Keywords}: optimal monetary policy, optimal simple rules,
stochastic dominance, stochastic calculus, DSGE model, strategy-proof,
Nash equilibrium, zero-knowledge, crypto-currency

~

\textbf{JEL classification}: C11, C54, D58, D81, E42, E47, E52, E61,
G11\\
\\
\\
\end{abstract}
\pagebreak{}

\tableofcontents{}

\pagebreak{}

\section{Introduction}

One of the notorious deficiencies of crypto-currencies is their lack
of monetary policy, as currently defined and studied in the field
of macroeconomics: nonetheless, monetary crypto-policymakers must
act in an optimal manner. In this paper, we initiate the study of
optimal monetary policies for crypto-currencies in order to derive
optimal simple rules that stochastically dominate the monetary policy
of other previous crypto-currencies, and ultimately, prove that the
efficient portfolio is to go long on the stochastically dominant crypto-currency.

\subsection{Contributions}

In summary, we make the following contributions:
\begin{itemize}
\item pioneer the introduction of the first optimal monetary policy for
crypto-currencies
\item devise the first stochastically dominant crypto-currency, its dominance
arising from its optimal monetary policy
\item derive optimal simple rules for a crypto-currency in a Dynamic Stochastic
General Equilibria model
\item prove that the efficient portfolio is to go long on stochastically
dominant crypto-currencies: in fact, it's a strategy-proof arbitrage
featuring a higher Omega ratio with a higher expected return, inducing
a Nash equilibrium over the crypto-currency market
\item describe how zero-knowledge proofs for the implemented monetary policy
are committed on the blockchain
\end{itemize}
In a nutshell, we contribute a new methodology for analysing and deriving
optimal simple rules for the monetary policy of stochastically dominant
crypto-currencies, in order to create efficient portfolios of stochastically
dominant crypto-currencies. This paper intends to be a self-contained
guide covering all the necessary theory and practical aspects.

In section \ref{sec:Related-Literature}, \textcolor{black}{we discuss
related literature and prior work. In section \ref{sec:Environment-Framework-Optimal-Portfolio},
we introduce our economic environment, analysis framework, and efficient
portfolio. In section \ref{sec:Model-Policies}, we describe our economic
model and optimal monetary policies. Finally, we detail some features
of the technical implementation in section \ref{sec:Implementation-Details},
including how to commit the implemented zero-knowledge policy, and
then we conclude in section \ref{sec:Conclusion}.}

The reader interested in less theoretic and most empirical analysis
may skip to subsection \textcolor{black}{\ref{subsec:Ranking-Policy-Rules},}

\section{\label{sec:Related-Literature}Related Literature}

The seminal contribution of this paper is to start the study of the
first optimal monetary policy for crypto-currencies: until now, all
the study of the field was concentrated on the monetary policies for
stablecoins \cite{stablecoinSurvey,cryptoeprint:2019:1054} or the
models of the interaction between crypto-currencies, fiat currencies
and/or CBDCs with a view of understanding their shocks to the economy
(starting from the seminal \cite{macroeconomicsCBDC}).

\begin{figure}[H]
\includegraphics[scale=0.15]{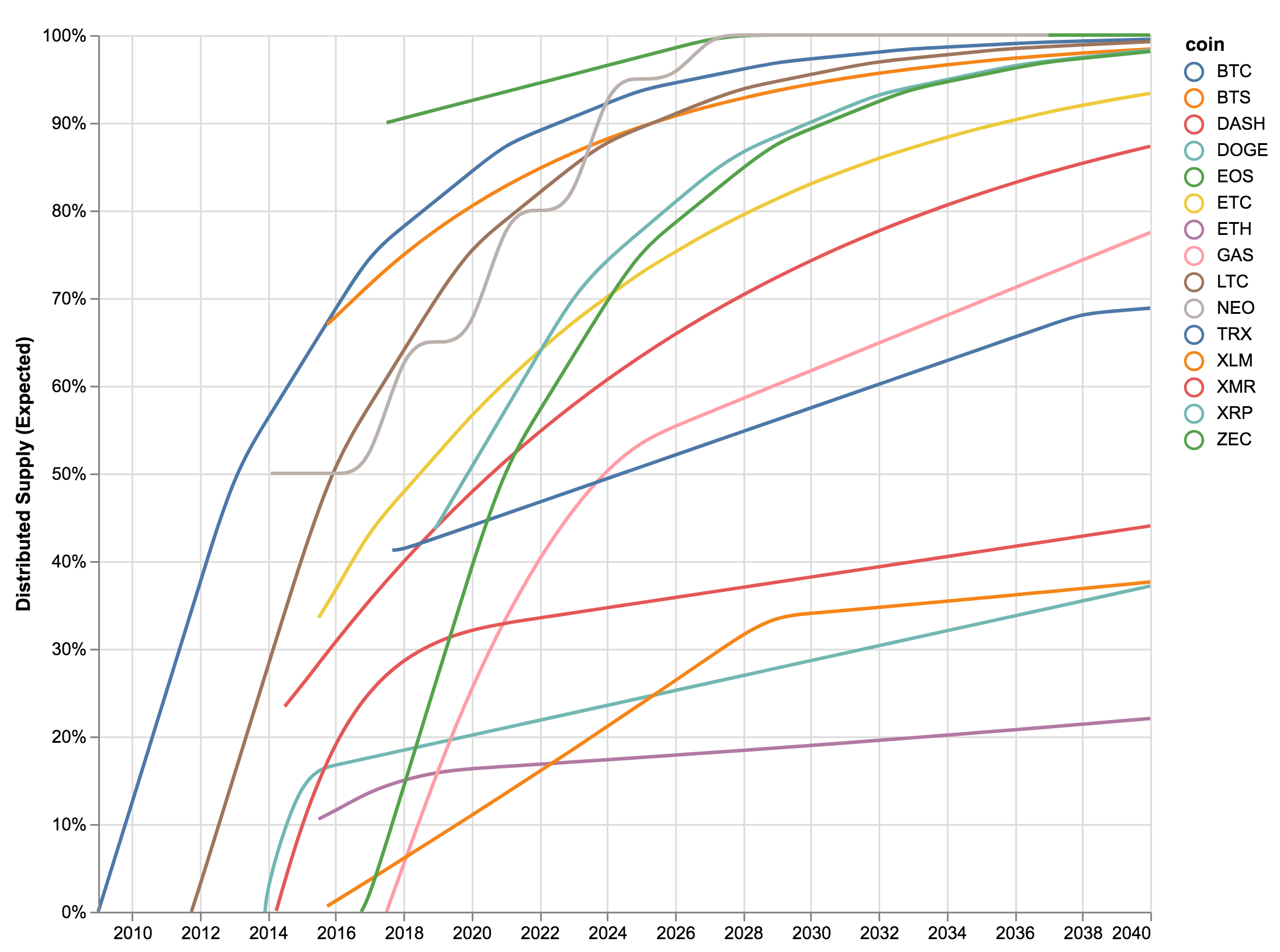}

\caption{\label{fig:Relative-supply-of-cryptocurrencies}Relative supply of
crypto-currencies\cite{cryptoMonetaryBase}}
\end{figure}

\begin{figure}[H]
\begin{centering}
\includegraphics[scale=0.25]{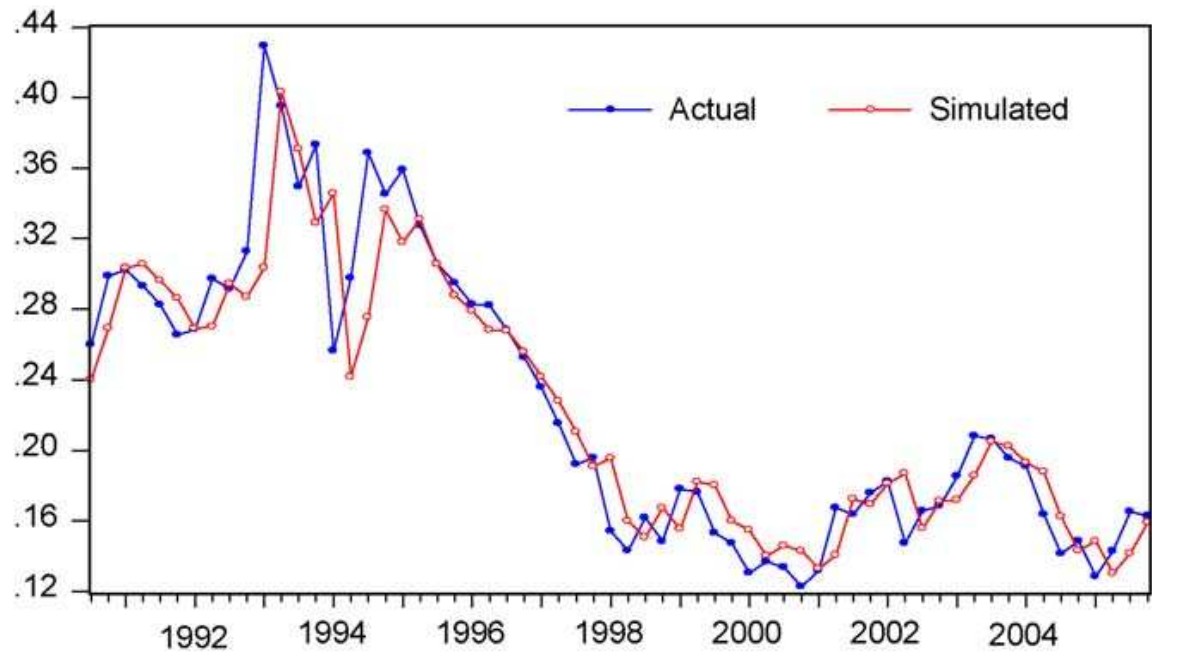}\caption{\label{fig:China's-simulated-quantity}China's simulated quantity
rule and actual M2 growth \cite{chinaPolicyRule}}
\par\end{centering}
\end{figure}

However, the study of the monetary policy of crypto-currencies has
always been relegated to matters related to their supply \cite{cryptoMonetaryBase},
as shown in the previous Figure \ref{fig:Relative-supply-of-cryptocurrencies},
in stark contrast with the quantity rules of money used in the real
world (e.g., China's policy rule as shown in the previous Figure \ref{fig:China's-simulated-quantity}).
In fact, the equational expression of both policy rules couldn't be
more different: Bitcoin's supply equation $S_{t}$ in period $t$
can be given by
\[
S_{t}^{BTC}=21\times10^{7}\times\left(1-\alpha^{t}\right)
\]
where $\alpha$ is the growth rate ($\alpha\approx0.825$ for yearly
periods and $\alpha\approx0.953$ for quarterly periods, see \ref{eq:bitcoinSupplyPeriod}),
and China's quantity rule of money \cite{chinaPolicyRule} for the
previous Figure \ref{fig:China's-simulated-quantity} can be given
by
\[
\Delta M_{t}=\Delta M_{t}^{*}+\theta_{1}\Delta M_{t-1}-\theta_{2}\hat{Y}_{t}-\theta_{3}\left(\pi_{t}-\pi_{t}^{*}\right)
\]
with $\Delta M$ denoting the nominal money growth, $\Delta M_{t}^{*}$
the log of equilibrium money growth, $\theta_{1}=0.88$ the lag of
nominal money growth, $\hat{Y_{t}}$ the output gap, $\theta_{2}=0.16$
the coefficient of response to changes of the output gap, $\pi_{t}$
the inflation rate in period $t$, $\pi_{t}^{*}$ the target inflation
rate, and $\theta_{3}=0.06$ the coefficient of response to changes
in inflation.

There is no related literature about what should be the optimal monetary
policy of a crypto-currency according to the methods of modern macro-economics,
as customarily practised on central banks: in fact, current crypto-currencies
are designed for \textit{anarcho-autarkic} settings on which they
don't have to keep track of inflation, GDP, or money growth, not even
the exchange rate of other crypto-currencies.

Moreover, previous sources of dominance in the crypto-currency market
were the first mover advantage of Bitcoin, or the network effects
inherent to payment networks \cite{cryptoCompetition} : as a novel
contribution, this paper introduces optimal monetary policies as a
source of dominance in the crypto-currency market.

Furthermore, simple rules are preferred over complex models \cite{simpleRulesAI}:
in the foreseeable future, simple rules will still dominate the design
of markets over complex models due to their many strengths and few
weaknesses \cite{taylorSimpleRules}.

\subsection{Comparison with prior work}

Previous work from the same author \cite{cryptoeprint:2019:1054}
described methods to conduct the monetary policy in a decentralised
fashion, but with the following differences:
\begin{enumerate}
\item Previous work \cite{cryptoeprint:2019:1054} focused on a stablecoin,
but this paper targets the volatile crypto-currency market.
\begin{enumerate}
\item However, this paper introduces the novelty of \textit{stochastic dominance}
of monetary policy rules and its usefulness to dominate other previous
crypto-currencies.
\end{enumerate}
\item The technical implementation of this paper is significantly simpler
than \cite{cryptoeprint:2019:1054}, without compromising the security
model: that is, it provides similar cryptographic guarantees on a
decentralised blockchain.
\end{enumerate}
Additionally, the results of this paper are also valid for the setting
of \cite{cryptoeprint:2019:1054} : the optimal simple policy rules
obtained in this paper (\ref{subsec:Ranking-Simple-Policy-Rules}
and \ref{subsec:Ranking-Policy-Rules}) could be directly incorporated
into the ``Economic Model for a Central-Banked Currency'' (Section
4.3 of \cite{cryptoeprint:2019:1054}).

\subsection{Survey of the Monetary Policy Impact on Crypto-currencies}

Although crypto-currencies such as Bitcoin were designed to replace
the discretionary decisions of monetary policymakers from central
banks, even to insulate them from macro-economic shocks, in reality
their decisions continue impacting their price and volatility. In
this subsection, a survey of recent research about this topic is presented,
which shall inform the design of monetary policy rules in the next
subsection (\ref{subsec:Ranking-Policy-Rules}):
\begin{flushleft}
\begin{table}[H]
\begin{raggedright}
\begin{tabular}{|c|c|>{\centering}p{0.65\textwidth}|}
\hline 
Paper & Period & Results\tabularnewline
\hline 
\hline 
\cite{monetaryShocksBitcoin} & 2010-2020 & Unanticipated 1 bp on 2-year Treasury yield is about a 0.25\% decrease
in Bitcoin price and 1.23\% three days later (stronger at high and
low quantiles)\tabularnewline
\hline 
\cite{monetaryPolicyBitcoin} & 2014-2021 & Disinflationary ECB policy shocks (2-year interest rates of 10 basis
points) lead to a persistent decrease in Bitcoin price (-20\%), whereas
inflationary ECB information shocks lead to price increases; conversely,
contractionary US policy shocks (2-year interest rates of 10 basis
point) increase Bitcoin prices (+7\%) but fall during expansionary
US information shocks (due to flows to foreign exchanges with emerging
market currencies)\tabularnewline
\hline 
\cite{qeBurtsBitcoin} & 2016-2021 & Cointegration between Bitcoin prices and M2, deeper with time delays\tabularnewline
\hline 
\cite{bitcoinNotGold} & 2010-2018 & SVAR model shows no response of Bitcoin prices to shocks to nominal
interest rate (1-year US treasury rate), only to stocks, VIX; but
increase after a positive shock to the price level (Billion prices
index)\tabularnewline
\hline 
\cite{cryptoFOMC} & 2013-2017 & Mineable crypto-currencies show US volatility spillovers during FOMC
announcement period, but not dApp or protocols\tabularnewline
\hline 
\cite{cryptoMacroFOMC} & 2010-2018 & Bitcoin price increases 0.26\% at no FOMC announcement, 0.96\% on
the day before and decreases 1\% on the announcement day. Bitcoin
price doesn't change on CPI, PPI or employment rate announcements\tabularnewline
\hline 
\cite{inflationCrypto} & 2010-2021 & Positive link between cryptocurrencies and forward inflation rates
is identified only during COVID-19\tabularnewline
\hline 
\cite{bitcoinInflationHedge} & 2010-2020 & Bitcoin prices appreciate against inflation (or inflation expectation)
shocks, but do not decrease after policy (1-year US treasury rate)
uncertainty shocks (i.e., only when excluding ZLB constraint)\tabularnewline
\hline 
\cite{inflationBitcoinDescriptiveAnalysis} & 2019-2020 & Daily changes in Bitcoin prices Granger cause changes in the forward
inflation rate in a significant and persistent way, but not vice-versa\tabularnewline
\hline 
\cite{exchangeRateBitcoin} & 2010-2014 & In the short term, Bitcoin price adjusts to changes in money supply,
GDP, inflation, and interest rate\tabularnewline
\hline 
\end{tabular}
\par\end{raggedright}
\caption{Survey of Monetary Policy Impact on Crypto-currencies}
\end{table}
\par\end{flushleft}

\begin{flushleft}
~\\
A practical example of the effects of inflation on Bitcoin price can
be found below, shedding \$1K on 13/9/2022 in just 3 minutes (10\%
of market capitalization) as US CPI inflation for August overshoots
at 8.3\% year-on--year (expected 8.1\%):
\par\end{flushleft}

\begin{figure}[H]
\centering{}\includegraphics[scale=0.35]{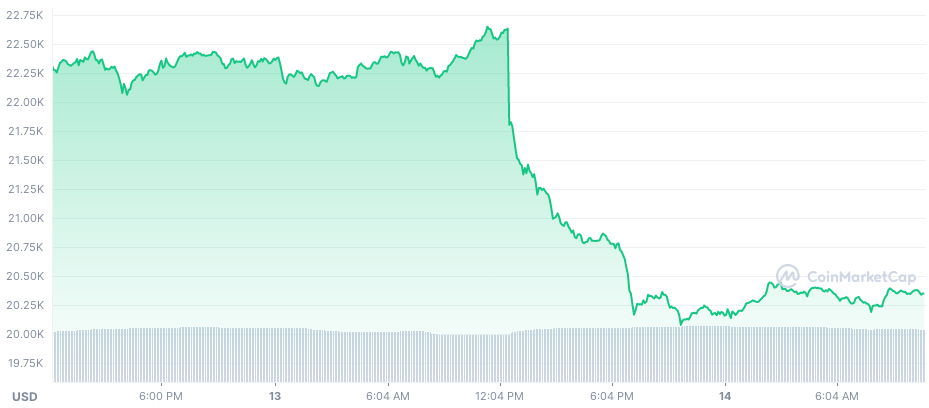}\caption{\label{fig:BTC_inflation_effect}Effect of inflation on BTC/USD}
\end{figure}

\section{\label{sec:Environment-Framework-Optimal-Portfolio}Environment,
Framework and Efficient Portfolio}

We consider pre-commitment rules in rational expectations models by
a Bayesian risk-averse policymaker that is given the task to choose
a policy feedback coefficient function mapping the parameter space
into the set of policy feedback coefficients interpreted as random
variables with probability distributions given from the posterior
distributions of the model parameters, in order to minimise the expected
disutilities of welfare loss for all disutility functions by ranking
the policy rules according to a stochastic dominance criterion that
is robust against all of the parameter uncertainty about the structure
of the economic model.

\subsection{\label{subsec:Economic-environment}Economic Environment}

The setting of this paper is the general form of linear rational expectations
models with uncertainty as set out in \cite{andersonLREM}: many Dynamic
Stochastic General Equilibrium models can be approximated by linear
rational expectation (LRE) equations,
\begin{multline}
\smash[b]{F_{1}\left(\theta_{1},\theta_{2}\right)\mathbb{E}_{t}x_{t+1}+F_{1}\left(\theta_{1},\theta_{2}\right)\mathbb{E}_{t}x_{t+1}+F_{2}\left(\theta_{1},\theta_{2}\right)\mathbb{E}_{t}u_{t+1}}\\
+F_{3}\left(\theta_{1},\theta_{2}\right)x_{t}+F_{4}\left(\theta_{1},\theta_{2}\right)u_{t}+F\left(\theta_{1},\theta_{2}\right)v_{t}=0,\label{eq:dynamicPrivateSector}
\end{multline}

\begin{align}
G_{2}\left(\phi\right)\mathbb{E}_{t}x_{t+1}+G_{3}\left(\phi\right)x_{t}+G\left(\phi\right)v_{t} & =G_{1}\left(\phi\right)u_{t},\label{eq:policyEquation}\\
M_{1}\left(\theta_{m}\right)x_{t}+M_{2}\left(\theta_{m}\right)u_{t}+M\left(\theta_{m}\right)v_{m,t} & =y_{t},t=0,1,2,\ldots\label{eq:measurementEquations}
\end{align}
with the equation \ref{eq:dynamicPrivateSector} describing the dynamics
of the private sector around the deterministic steady-state, \ref{eq:policyEquation}
being the policy equation, and \ref{eq:measurementEquations} the
measurement equations, all the above using the following notation:
\begin{itemize}
\item $x_{t}$ is a vector of $n$ non-policy endogenous variables
\item $u_{t}$ a vector of $k$ policy variables
\item $y_{t}$ a vector of $m\leq n+k$ observable variables
\item $\mathbb{E}_{t}$ is the operator of conditional expectation with
respect to an information set in period $t$ 
\item $F_{i},F,G_{i},G,M_{i},$ and $M$ are matrices depending on the parameters
\item $v_{m,t}$ and $v_{t}$ are vectors of independent and identically
distributed innovations with zero mean and identity variance-covariance
matrix $I$ 
\item $\theta_{1}$ is a vector of structural non-policy random parameters
\item $\theta_{2}$ is a vector of structural non-policy calibrated parameters
\item $\theta_{m}$ is a vector of measurement parameters
\item $\phi$ is a vector of random policy parameters (feedback or response
coefficients)
\end{itemize}
The solution to the system of linear rational expectation equations
\ref{eq:dynamicPrivateSector} - \ref{eq:policyEquation} is given
by a state equation of the form
\begin{equation}
z_{t}=A\left(\theta_{s},\phi\right)z_{t-1}+B\left(\theta_{s},\phi\right)v_{t},\,\,\,\,t=1,2,\ldots,\label{eq:solutionEquations}
\end{equation}
for the initial vector of states variables $z_{0}=\left[x_{0}^{'},u_{0}\right]^{'}$
and for unknown matrices $A\left(\theta_{s},\phi\right)$, $B\left(\theta_{s},\phi\right)$
with $\theta_{s}=\left[\theta_{1}^{'},\theta_{2}^{'}\right]^{'}$.

The parameter uncertainty in model \ref{eq:measurementEquations}
- \ref{eq:solutionEquations} is measured by the posterior probability
distribution function according to the following Bayes rule:
\begin{equation}
p\left(\theta,\phi|Y_{t}\right)=\frac{p\left(\theta,\phi\right)p\left(Y_{t}|\theta,\phi\right)}{p\left(Y_{t}\right)},\label{eq:bayesRule}
\end{equation}
where $\theta=\left[\theta_{s},\theta_{m}^{'}\right]^{'}$, $Y_{t}=\left[y_{1}^{'},y_{2}^{'},\ldots,y_{t}^{'}\right]^{'}$
is a sequence of observable vectors at time $t$, $p\left(Y_{t}|\theta,\phi\right)$
is a likelihood function, and $p\left(\theta,\phi\right)=p\left(\theta\right)p\left(\phi\right)$
is a prior posterior probability distribution function. The elements
of $A\left(\theta_{s},\phi\right)$, and $B\left(\theta_{s},\phi\right)$are
usually non-linear functions of the vectors $\theta_{s}$ and $\phi$,
and the posteriors are not analytically available so we use the likelihood
principle to treat posteriors as a measure of uncertainty about the
parameters; thus, simulations are used to find approximations to the
marginal posterior distributions $\theta$ and $\phi$.

With this approach, optimal policy coefficients are assumed to be
random variables with probability distributions inherited from the
posterior distributions of the structural model parameters observed
with parameter uncertainty, avoiding treating optimal feedback coefficients
as fixed numbers much like policymakers usually do.

\subsection{\label{subsec:Decision-framework}Decision Framework}

In this subsection, we use a decision procedure to evaluate and rank
simple policy rules in rational expectation models. A Bayesian policymaker
formulates a statistical decision problem to choose a policy rule
under parameter uncertainty with the following tuple
\[
\left(Y_{t},p\left(\theta\right),p\left(\phi_{l}\right),\Theta,D,M,L_{t}\right)
\]
in which each term defines:
\begin{itemize}
\item $Y_{t}$ denotes the history of the observable variables over $t$
periods
\item subjective prior distributions $p\left(\theta\right)$ for the structural
parameters $\theta=\left[\theta_{s}^{'},\theta_{m}^{'}\right]^{'}\in\varTheta=\varTheta_{s}\times\varTheta_{m}$ 
\item subjective prior distributions $p\left(\phi_{l}\right)$for the policy
parameters $\phi_{l}\in\varPhi_{l}$ for $l=1,2,\ldots,N$
\item a set $D$ of actions
\item a set $M=\left\{ P_{z|\theta_{s},d}:\theta_{s}\in\varTheta_{s},d\in D\right\} $
of linear rational expectations models under consideration, differing
in the values of the structural parameters $\theta_{s}\in\Theta_{s}$
and the values of the policymaker's action $d\in D$
\item loss function $L_{t}\left(\theta,d\right)$ that quantifies the policymaker's
choice of applying a given policy rule when a particular model holds
\end{itemize}
Considering how to rank a set of policy rules in the model of \ref{subsec:Decision-framework}
with $N\geq2$ different functional forms
\begin{equation}
G_{1l}\left(\phi_{l}\right)u_{l,t}=G_{2l}\left(\phi_{l}\right)\mathbb{E}_{t}x_{t+1}+G_{3l}\left(\phi_{l}\right)x_{t}+G_{l}\left(\phi_{l}\right)v_{t},\,\,\,\,t=0,1,2,\ldots\label{eq:rankingPolicyRules}
\end{equation}
with $l=1,2,\ldots,N$, the vector $\phi_{l}\in\varPhi_{l}$ collecting
the policy feedback coefficients, with $G_{il}$ and $G_{l}$ being
matrices that depend on the policy feedback parameters.

The decision space $D$ is of the form
\begin{align}
D & =\left\{ \left(l,f_{l}\right):l=1,2,\ldots,N\right.\\
 & \left.f_{l}:\varTheta\rightarrow\varPhi_{l}\right\} 
\end{align}
in which the following conditions hold:
\begin{itemize}
\item the admissible policies $d=\left(l,f_{l}\right)$ is a rule from \ref{eq:rankingPolicyRules}
\item $f_{l}:\varTheta\rightarrow\varPhi_{l}$ is a policy feedback coefficient
from a given class of measurable functions $F_{l}$ such that the
system of linear rational expectation equations \ref{eq:dynamicPrivateSector}
and \ref{eq:rankingPolicyRules}, with $\phi_{l}=f_{l}\left(\theta\right)$
for all $\theta\in\varTheta$ has a solution which is given by the
state equation
\begin{equation}
z_{l,t}=A_{l}\left(\theta_{s},\phi_{l}\right)z_{l,t-1}+B_{l}\left(\theta_{s},\phi_{l}\right)v_{t},\,\,\,\,\,t\geq1\label{eq:stateEquationDecisionSpace}
\end{equation}
with $z_{0}$ being an initial state, $z_{l,t}=\left[x_{t}^{'},u_{l,t}^{'}\right]^{'}$,
$z_{l,0}=z_{0}$, and $A_{l}\left(\theta_{s},\phi_{l}\right),B_{l}\left(\theta_{s},\phi_{l}\right)$
are unknown matrices with $\theta_{s}=\left[\theta_{1}^{'},\theta_{2}^{'}\right]^{'}$
\item every set $F_{l}$ includes all constant functions $f_{l}\left(\theta\right)=\text{const}$
\item the parameters space $\varPhi_{l}$ consists of all vectors of policy
parameters $\phi_{l}$ for every $l=1,2,\ldots$ such that for all
$\theta_{s}\in\varTheta_{s}$, the system of linear rational expectation
equations \ref{eq:dynamicPrivateSector} and \ref{eq:rankingPolicyRules}
has a unique solution
\end{itemize}
The procedure of the Bayesian policymaker is to first observe the
history of the observable variables $Y_{t}$ over $t$ periods, and
for every $l=1,2,\ldots,N$ sets the subjective prior distributions
$p\left(\theta\right)$ of the structural parameters and $p\left(\phi_{l}\right)$
of the policy parameters $\phi_{l}\in\varPhi_{l}$. Then, it analyses
the following set of linear rational expectation models,
\[
M=\left\{ P_{z|\theta_{s},d}:\theta_{s}\in\varPhi_{s},d\in D\right\} 
\]
for endogenous non-policy $x_{t}$ described by \ref{eq:dynamicPrivateSector}
and policy variables $u_{l,t}$ described by \ref{eq:rankingPolicyRules}.
The predictive probability distribution $P_{z|\theta_{s},d}$ of the
future state variables $z=\left(z_{l,t+s}\right)_{s=0,1,2,\ldots}\in\mathbb{Z}$
evolves according to \ref{eq:stateEquationDecisionSpace}.

\subsubsection{Welfare Loss Function}

The welfare loss function of the Bayesian decision maker's objective
at time $t$ is defined by
\[
L_{t}:Z_{t}\times V\times\varTheta\times D\rightarrow\left[0,\infty\right)
\]
receiving the following parameters:
\begin{itemize}
\item a vector of current state variables $z_{t}\in Z_{t}$
\item all future shocks $v=\left\{ \left(v_{t+s},v_{m,t+s}\right)\right\} _{s=1,2,\ldots}\in V$ 
\item all vectors of structural parameters $\theta$ from the parameter
space $\varTheta$ 
\item all admissible decisions $d$ from the decision space $D$ 
\end{itemize}
In order to evaluate the objective function, the Bayesian policymaker
could take the unconditional average of the welfare losses over the
current state and all possible future shocks:
\begin{equation}
L_{t}\left(\theta,d\right)=\int_{Z_{t}}\int_{V}L_{t}\left(z_{t},v,\theta,d\right)dP_{v}\left(v\right)dP_{z_{t}}\left(z_{t}\right)
\end{equation}
or the conditional expected value of the welfare loss given $z_{t}$:
\begin{equation}
L_{t}\left(\theta,d|z_{t}\right)=\int_{V}L_{t}\left(z_{t},v,\theta,d\right)dP_{v}\left(v\right)
\end{equation}
with $\theta=\left[\theta_{s}^{'},\theta_{m}^{'}\right]^{'}\in\varTheta$
and the policymaker decision $d\in D$ that is able to modify the
model structure $P_{z|\theta_{s},d}$, the posterior distribution
of the structural parameters $P_{\theta|Y_{t},d}$, and the value
of the expected welfare loss $L_{t}\left(\theta,d\right)$.

In this paper, we will use a quadratic welfare loss function given
by:
\begin{equation}
L_{t}\left(\theta,d\right)=\text{tr}\left(W\sum_{z_{l}}\left(\theta_{s},d\right)\right)=\text{var}_{\theta_{s},d}\left(\hat{\pi}_{t}\right)+w_{y}\text{var}_{\theta_{s},d}\left(\hat{y}_{t}\right)\label{eq:welfareLossFunction}
\end{equation}
for all $d=\left(l,f_{l}\right)\in D$ and $\theta_{s}\in\Theta_{s}$
where $\text{var}_{\theta,d}\left(z_{l,t,i}\right)$ is the unconditional
variance of the state variable $z_{l,t,i}$ in the $l$-th specification
of the DSGE model, while $w_{y}$ is the diagonal weight of $W$ for
the output gap. These diagonal weights reflect the monetary policy
preferences of the central bank over the objectives: specifically,
we set $w_{y}=0.05$ and $w_{\pi}=1$.

\subsubsection{\label{subsec:Ranking-Simple-Policy-Rules}Ranking Simple Policy
Rules}

Now we can start formulating robust optimality criteria to generate
rankings of simple policy rules \ref{eq:rankingPolicyRules} based
on the optimal Bayesian policymaker's objective function $L_{t}$:
first, a fixed vector of structural parameters $\hat{\theta}\in\varTheta$
is chosen from the parameter space $\varTheta$, and then the expected
welfare loss $L_{t}\left(\hat{\theta},d\right)$ is minimised subject
to the recursive state equations \ref{eq:stateEquationDecisionSpace}
under criteria of $k$-degree stochastic dominance (SD$k$) \cite{levyStochasticDominance}.
Stochastic dominance is a useful concept for analysing risky decision-making
under uncertainty when only partial information about the decision
maker's risk preferences is available.
\begin{defn}
(SD$k$ ordering) \label{def:(SDk-ordering)}. Defined by the indefinitely
many inequalities
\begin{equation}
\int_{0}^{L^{*}}u\left(x\right)dF_{L_{1}}\left(x\right)\leq\int_{0}^{L^{*}}u\left(x\right)dF_{L_{2}}\left(x\right)\label{eq:sdk-ordering}
\end{equation}
between the expected disutilities of non-negative valued random losses
$L_{1}\leq_{\text{SD}k}L_{2}$ with the cumulative distribution functions
$L_{1}\sim F_{L_{2}}$ and $L_{2}\sim F_{L_{2}}$ , for all functions
$u\in U_{k}$ with strict inequality for some $u$, where $U_{k}$
is the set of all disutility functions with the $i$-th derivative
of $u$ such that $u'\geq0$, $u''\geq0,\ldots,u^{\left(k\right)}\geq0$.
Note that SD$k$ implies SD$l$ for all $k>l$, and that the SD$\infty$
dominance of $L_{2}$ over $L_{1}$ implies that $L_{1}\leq_{\text{SD}k}L_{2}$
holds for some finite $k$. Additionally, we denote with $\leq_{\text{SD}k}$
the inequality between random welfare losses defined by the $\text{SD}k$
ordering for $k=1,2,\ldots,\infty$. Recall that SD1 ordering assumes
all non-decreasing disutility functions (non-satiable); SD2 is for
risk-averse policymakers towards welfare losses, restricting the disutility
functions to convex and non-decreasing; SD3 additionally prefers negatively
skewed welfare loss distributions (prudence); and SD4 requires that
$u^{4}\leq0$ (temperance).
\end{defn}

Accordingly, simple policy rules can be analysed using the SD$k$
ordering from the previous Definition \ref{def:(SDk-ordering)}.
\begin{defn}
(SD$k$-optimal policy) \label{def:(SDk-optimal-policy).}. Finding
the best SD$k$-optimal simple policy under parameter uncertainty
is solved by searching for the SD$k$-optimal decision $d_{1}^{\text{SD}k}=\left(l_{1}^{\text{SD}k},f_{l_{1}^{\text{SD}k}}\right)\in D$
from the set of all admissible decisions $D$ such that the corresponding
distribution of welfare loss $L_{t}\left(\cdot,d_{1}^{\text{SD}k}\right)$
satisfies
\begin{equation}
L_{t}\left(\cdot,d_{1}^{\text{SD}k}\right)\leq_{\text{SD}k}L_{t}\left(\cdot,d\right)\label{eq:sdk-optimal-policy}
\end{equation}
for all $d\in D$ and subject to \ref{eq:stateEquationDecisionSpace}:
$d_{1}^{\text{SD}k}$ generates the distribution of minimised welfare
loss, the smallest in terms of the SD$k$ ordering.
\end{defn}

Assume that the Bayesian policymaker solves a parameterised optimisation
problem to find the value of the optimal policy feedback coefficient
function $f_{l}^{\text{min}}\left(\theta\right)$:
\begin{equation}
L_{l,t}^{\text{min }}\left(\theta\right)=\underset{\phi_{l}\in\varPhi_{l}}{\text{min}}L_{t}\left(\theta,\left(l,\phi_{l}\right)\right)=L_{t}\left(\theta,\left(l,f^{\text{min }}\left(\theta\right)\right)\right)\label{eq:parametrisedOptimisationProblem}
\end{equation}
for each value of the structural parameters $\theta\in\varTheta$
and for every policy specification $l=\left\{ 1,2,\ldots,N\right\} $
given in \ref{eq:rankingPolicyRules}. Note that $f_{l}^{\text{min}}\left(\theta\right)$
is a selection from the optimal choice correspondence set:

\[
f_{l}^{\text{min}}\left(\theta\right)\in\varPhi_{l}^{\text{min }}\left(\theta\right)=\left\{ \phi_{l}^{\text{min}}\in\varPhi_{l}:L_{l,t}^{\text{min}}\left(\theta\right)=L_{t}\left(\theta,\left(l,\phi_{l}^{\text{min}}\right)\right)\right\} 
\]
We define $\phi_{l}^{\text{min}}=f_{l}^{\text{min}}\left(\theta\right)$
to be the vector of optimal policy feedback coefficient of rule $l$
calculated for the vector of structural parameters $\theta$: thus,
the optimal policy feedback coefficient function $f_{l}^{\text{min}}:\varTheta\rightarrow\varPhi_{l}$
is measurable and the pair $\left(l,f_{l}^{\text{min}}\right)$ belongs
to $D$.

The uncertainty of the structural parameters is considered in order
to find the probability distribution of the optimal policy response
coefficients 

\begin{equation}
\phi_{l}^{\text{min}}\sim p_{\theta}\left(\left(f_{l}^{\text{min}}\right)^{-1}|Y_{t},l\right)\label{eq:optimalPolicyResponseCoefficients}
\end{equation}
and the minimised welfare loss is given by

\begin{equation}
L_{l,t}^{\text{min}}\sim p_{\theta}\left(\left(L_{t}\circ f_{l}^{\text{min}}\right)^{-1}|Y_{t},l\right)\label{eq:minimizedWelfareLoss}
\end{equation}
where the inverse image of $A\in\mathcal{B\left(\varPhi\right)}$
under $\theta\rightarrow f_{l}^{\text{min}}\left(\theta\right)$ is
\[
\left(f_{l}^{\text{min}}\right)^{-1}\left(A\right)=\left\{ \theta\in\varTheta:f_{l}^{\text{min}}\left(\theta\right)\in A\right\} 
\]
and the inverse image of $B\in\mathcal{B\left(\varPhi\right)}$ under
$\theta\rightarrow L_{t}\left(\theta,f_{l}^{\text{min}}\left(\theta\right)\right)$
is
\[
\left(L_{t}\circ f_{l}^{\text{min}}\right)^{-1}\left(B\right)=\left\{ \theta\in\varTheta:L_{t}\left(\theta,f_{l}^{\text{min}}\left(\theta\right)\right)\in B\right\} 
\]

Next theorem \ref{eq:theo-finding-sdk} gives sufficient conditions
for the optimal solution to \ref{eq:sdk-optimal-policy} and shows
how the SD$k$-optimal decision can be found.
\begin{thm}
\label{thm:finding-sdk}Assume that the decision sets $\varPhi_{l}$
for $l=1,2,\ldots,N$ are non-empty compact subsets of $\mathbb{R}^{r}$,
the parameter space $\varTheta$ is an open subset of $\mathbb{R}^{p}$,
and the policymaker's welfare loss $L_{t}$ is a Carathéodory-type
integrable function. If $d^{*}=\left(l^{*},f_{l^{*}}\right)\in D$
is a policymaker decision such that $l^{*}$ is defined by 
\begin{equation}
L_{l^{*},t}^{\text{min}}\leq_{\text{SD}k}L_{l,t}^{\text{min}},\,\,\,\,\forall l\in\left\{ 1,2,\ldots,N\right\} \label{eq:theo-finding-sdk}
\end{equation}
where $L_{l,t}^{\text{min}},l\in\left\{ 1,2,\ldots,N\right\} $ are
random minimised welfare losses as defined in \ref{eq:parametrisedOptimisationProblem},
\ref{eq:optimalPolicyResponseCoefficients} and \ref{eq:minimizedWelfareLoss};
and $f_{l^{*}}=f_{l^{*}}^{\text{min }}$ is the optimal policy feedback
coefficient function that solves \ref{eq:parametrisedOptimisationProblem},
as denoted by
\begin{equation}
\underset{\phi_{l^{*}}\in\varPhi_{l^{*}}}{\text{min}}L_{t}\left(\theta,\left(l^{*},\phi_{l^{*}}\right)\right)=L_{t}\left(\theta,\left(l^{*},f_{l^{*}}\left(\theta\right)\right)\right)=L_{l^{*},t}^{\text{min }}\left(\theta\right),\,\,\,\,\forall\theta\in\varTheta
\end{equation}
then $d^{*}=d_{1}^{\text{SD}k}$ is the SD$k$-optimal decision.
\end{thm}

\subsection{Efficient Portfolio}

In this subsection, we derive an efficient portfolio for stochastically
dominant crypto-currencies providing the best expected returns in
comparison with the other crypto-currencies. Instead of using the
mean-variance framework, we prefer to use marginal conditional stochastic
dominance \cite{mcsd}: all risk-averse investors prefer a portfolio
$A$ over a portfolio $B$ if the portfolio return of $A$ is stochastically
dominant over that of $B$, moving out all dominated assets. Furthermore,
almost marginal conditional stochastic dominance \cite{amcsd} could
be used to prevent extreme utility functions in the set of risk-averse
investors.

\subsubsection{Pricing Stochastic Dominance}

In order to ease exposition, suppose there are only two crypto-currencies
in two separate, segmented markets: a stochastically dominant crypto-currency,
$D$, and Bitcoin, $B$ (resp. any other PoW/PoS crypto-currency).
Consumption in the dominant market at time $t$ is denoted by $c_{t}^{D}$,
and $c_{t}^{B}$ in the Bitcoin denominated market (resp. any other
PoW/PoS crypto-currency). Consumers can transact in one market but
not both simultaneously, that is, utility $u_{j}\left(C_{j,t}\cdot ms_{j,t}\right)$
at time $t$ in market $j\in\left\{ D,B\right\} $, where $C_{j,t}$
denotes complete-market consumption that is distorted by the non-hedgeable
monetary policy shock $ms_{j,t}$, rendering incomplete the system
of markets: moreover, we assume that $ms_{j,t}$ and $C_{j,t}$ are
statistically independent for any $j\in\left\{ D,B\right\} $.
\begin{lem}
Suppose the utility function $u\left(x\right)$ is of the form Constant
Relative Risk-Aversion (CRRA), then the Stochastic Discount Factor
(SDF) in the dominant market is given by
\begin{equation}
M_{t+1}^{D}=\left(\frac{ms_{D,t+1}C_{D,t+1}}{ms_{D,t}C_{D,t}}\right)^{-\gamma}=M_{t+1}^{C_{D}}M_{t+1}^{ms_{D}}
\end{equation}
and in the Bitcoin market (resp. any other PoW/PoS crypto-currency)
is given by 
\begin{equation}
M_{t+1}^{B}=\left(\frac{ms_{B,t+1}C_{B,t+1}}{ms_{B,t}C_{B,t}}\right)^{-\gamma}=M_{t+1}^{C_{B}}M_{t+1}^{ms_{B}}
\end{equation}
where $\gamma$ is the coefficient of relative risk aversion, and
\[
M_{t+1}^{C_{j}}=\left(\frac{C_{j,t+1}}{C_{j,t}}\right)^{-\gamma},M_{t+1}^{ms_{j}}=\left(\frac{ms_{j,t+1}}{ms_{j,t}}\right)^{-\gamma}
\]
\end{lem}

~
\begin{lem}
\label{prop:(Fundamental-Pricing-Equation)}\textbf{\textup{(Fundamental
Pricing Equation).}} The Euler equation for holders of the stochastically
dominant crypto-currency is given by:
\begin{equation}
E_{t}\left[M_{t+1}^{B}\frac{Q_{t+1}}{Q_{t}}\right]=E_{t}\left[M_{t+1}^{D}\frac{M_{t+1}^{ms_{B}}}{M_{t+1}^{ms_{D}}}\right]=\frac{1}{R_{t+1}^{D}}
\end{equation}
and the Euler equation for the Bitcoin holder (resp. any other PoW/PoS
crypto-currency) is given by:
\begin{equation}
E_{t}\left[M_{t+1}^{D}\frac{Q_{t+1}}{Q_{t}}\right]=E_{t}\left[M_{t+1}^{B}\frac{M_{t+1}^{ms_{D}}}{M_{t+1}^{ms_{B}}}\right]=\frac{1}{R_{t+1}^{B}}
\end{equation}
where $Q_{t}$is the real exchange rate, and $R_{t+1}^{D}$ and $R_{t+1}^{B}$
are the risk-free rate in the dominant and Bitcoin market, respectively.
\end{lem}

We assume that the logarithm of the Stochastic Discount Factors (SDFs)
in the two markets are normally distributed: $m^{D}$, $m^{B}$, $m^{ms_{D}}$,
$m^{ms_{B}}$, $m^{C_{D}}$ and $m^{C_{B}}$ (i.e., we denote logarithms
of capitalised variables with their lowercase variant).
\begin{lem}
\label{prop:arbitrage-free-return}The arbitrage-free expected return
on the stochastically-dominant crypto-currency is given by:
\begin{equation}
E_{t}\left(\Delta q_{t+1}\right)=r_{t}^{B}-r_{t}^{D}+\frac{1}{2}\left[\text{Var}_{t}\left(m_{t+1}^{B}\right)-\text{Var}_{t}\left(m_{t+1}^{D}\right)\right]+E_{t}m_{t+1}^{ms_{B}}-E_{t}m_{t+1}^{ms_{D}}
\end{equation}
thus a relative rise of $E_{t}m_{t+1}^{ms_{B}}$ over $E_{t}m_{t+1}^{ms_{D}}$
leads to the appreciation of the stochastically dominant crypto-currency.
\end{lem}

\begin{defn}
\label{def:Logarithmic-utility-function}(Logarithmic utility function).
The utility function is
\[
u\left(x\right)=\log\left(x\right)
\]
thus we have the following additively separable representation for
the two shocks, consumption and monetary (resp. $cs$ and $ms$):
\[
u\left(cs\cdot ms\right)=\log\left(cs\cdot ms\right)=u_{1}\left(cs\right)\cdot u_{2}\left(ms\right)=\log\left(cs\right)+\log\left(ms\right)
\]
with values for $D_{1}$ and $B_{1}$ for $u_{1}$ and $D_{2}$ and
$B_{2}$ for $u_{2}$: $D_{1}$ and $D_{2}$are marginals of the joint
probability distribution $D\left(cs,ms\right)$ in the dominant market,
while $B_{1}$ and $B_{2}$ are marginals of the joint probability
distribution $B\left(cs,ms\right)$ in the Bitcoin market (resp. any
other PoW/PoS crypto-currency).
\end{defn}

\begin{lem}
\label{prop:FSD-dominant}\textbf{\textup{(First-order Stochastic
Dominance).}} A necessary and sufficient condition for the first-order
stochastic dominance of the stochastically dominant crypto-currency
over Bitcoin (resp. any other PoW/PoS crypto-currency) is 
\[
B_{1}\left(cs\right)\geq D_{1}\left(cs\right)
\]
and
\[
B_{2}\left(cs\right)\geq D_{2}\left(cs\right)
\]
with strong inequality for at least some values in $cs$ or in $ms$.
\end{lem}

Given the definitions \ref{def:Logarithmic-utility-function}, the
stochastically dominant crypto-currency in the dominant market with
a joint probability distribution $D\left(cs,ms\right)$, is preferred
to the Bitcoin market with $B\left(cs,ms\right)$ if:
\begin{align}
 & E_{D}u\left(cs,ms\right)-E_{B}u\left(cs,ms\right)\\
 & =\int\left(B_{1}\left(t\right)-D_{1}\left(t\right)\right)du_{1}\left(t\right)+\int\left(B_{2}\left(s\right)-D_{2}\left(s\right)\right)du_{2}\left(s\right)\\
 & =\int\left(B_{1}\left(t\right)-D_{1}\left(t\right)\right)\frac{1}{t}du\left(t\right)+\int\left(B_{2}\left(s\right)-D_{2}\left(s\right)\right)\frac{1}{s}du\left(s\right)\geq0
\end{align}

\begin{lem}
\label{prop:SSD-dominant}\textbf{\textup{(Second-order Stochastic
Dominance).}} A necessary and sufficient condition for the second-order
stochastic dominance of the stochastically dominant crypto-currency
over Bitcoin (resp. any other PoW/PoS crypto-currency)
\[
\int_{t}^{\infty}\left(B_{1}\left(s\right)-D_{1}\left(s\right)\right)ds\geq0
\]
and
\[
\int_{t}^{\infty}\left(B_{2}\left(s\right)-D_{2}\left(s\right)\right)ds\geq0
\]
for all $s>t$ with at least one strict inequality.
\end{lem}

We assume that the consumption shocks in the two markets, dominant
and Bitcoin, are roughly the same,
\[
D_{1}\approx B_{1}
\]
as both crypto-currencies are part of the same general economy (i.e.,
$cs$ is rendered $C$ as in its initial definition), thus monetary
shocks play a pivotal role: users prefer the stochastically dominant
crypto-currency given the projected expected utility of its stochastically
dominant monetary policy, $D_{2}$, over Bitcoin's $B_{2}$ (resp.
any other PoW/PoS crypto-currency).
\begin{thm}
\label{thm:dominantHigherPrice}\textbf{\textup{(Dominant Expected
Returns).}} A dominance relationship between the distribution of the
monetary policy shock of the stochastically dominant crypto-currency,
$ms_{D}\left(D_{2}\right)$, over Bitcoin, $ms_{B}\left(B_{2}\right)$,
implies a rise $E_{t}\left(\Delta q_{t+1}\right)$ \ref{prop:arbitrage-free-return}
of the price of the stochastically dominant crypto-currency $D$ in
terms of Bitcoin (resp. any other PoW/PoS crypto-currency).
\end{thm}

\begin{proof}
When $D_{2}$ dominates $B_{2}$ in the first-order sense \ref{prop:FSD-dominant},
then
\[
E_{t}m_{t+1}^{ms_{B}}-E_{t}m_{t+1}^{ms_{D}}>0
\]
and vice versa. Thus,
\begin{align*}
E_{t}m_{t+1}^{ms_{B}}-E_{t}m_{t+1}^{ms_{D}} & =ms_{B,t}E_{t}\left(\frac{1}{ms_{B,t+1}}\right)-ms_{D,t}E_{t}\left(\frac{1}{ms_{D,t+1}}\right)\\
 & =\int\frac{1}{t}dB_{2}\left(t\right)-\int\frac{1}{t}dD_{2}\left(t\right)\\
 & =\int\frac{1}{t}d\left(B_{2}-D_{2}\right)\\
 & =\int\left(D_{2}-B_{2}\right)d\frac{1}{t}\\
 & =-\int\left(D_{2}-B_{2}\right)\frac{1}{t\text{\texttwosuperior}}dt>0
\end{align*}
and the consequent rise of $E_{t}\left(\Delta q_{t+1}\right)$ as
given by \ref{prop:arbitrage-free-return}.
\end{proof}

\subsubsection{Efficient Stochastically Dominant Portfolio}

The stochastic dominance between two crypto-currencies of \ref{prop:FSD-dominant}
and \ref{prop:SSD-dominant} further extends into a stochastically
dominant portfolio of crypto-currencies. As in previous sections,
we can distinguish between different orders of portfolio dominance:
first (Definition \ref{def:FSD-Portfolio-Dominance}), second (Definition
\ref{def:SSD-Portfolio-Dominance}), ..., SD$k$-orders (Definition
\ref{def:sdk-Portfolio--dominance}) of portfolio dominance. Note
that recent empirical research corroborates that the inclusion of
crypto-currencies in portfolios is itself stochastically dominant
\cite{ASDcryptoFactorPortfolio,bitcoinStochasticSpanning,bootstrapStochasticBitcoin,tradingCryptoSSD,invBehaviorCrypto,diversificationCryptoSD}.

Given $N$ alternatives and a random vector of their outcomes $\varrho$,
a decision maker can combine them into portfolios and all portfolio
possibilities are denoted by
\[
\Lambda=\left\{ \lambda\in\mathbb{R}^{N}|1'\lambda=1,\,\,\,\,\lambda_{n}\geq0,\,\,\,\,n=1,2,\ldots,N\right\} 
\]

\begin{defn}
\label{def:FSD-Portfolio-Dominance}Portfolio $\lambda\in\varLambda$
dominates portfolio $\tau\in\varLambda$ by the first-order stochastic
dominance ($\varrho'\lambda\leq_{SD1}\varrho'\tau$) if
\[
F_{\varrho'\lambda}\left(x\right)\leq_{SD1}F_{\varrho'\tau}\left(x\right),\,\,\,\,\,\forall x\in\mathbb{R}
\]
with strict inequality for at least one $x\in\mathbb{R}$ and with
$F_{\varrho'\lambda}\left(x\right)$ denoting the cumulative probability
distribution of returns of portfolio $\lambda$. Necessary and sufficient
conditions for the first-order stochastic dominance ($\varrho'\lambda\leq_{SD1}\varrho'\tau$)
if:
\end{defn}

\begin{itemize}
\item $Eu\left(\varrho'\lambda\right)\geq Eu\left(\varrho'\tau\right)$
for all expected utility ($Eu$) functions and strict inequality holds
for at least some utility function
\item $F_{\varrho'\lambda}^{-1}\left(y\right)\leq F_{\varrho'\tau}^{-1}\left(y\right)$
for all $y\in\left[0,1\right]$ with strict inequality for at least
one $y\in\left[0,1\right]$
\item $VaR_{\alpha}\left(-\varrho'\lambda\right)\leq VaR_{\lambda}\left(-\varrho'\tau\right)$
for all $\alpha\in\left[0,1\right]$ with strict inequality for at
least one $\alpha\in\left[0,1\right]$
\end{itemize}
\begin{defn}
\label{def:FSD-efficient} A given portfolio $\tau\in\varLambda$
is first-order stochastic dominant (Definition \ref{def:FSD-Portfolio-Dominance})
inefficient if there exists portfolio $\lambda\in\varLambda$ such
that $\varrho\text{'}\lambda\leq_{SD1}\varrho\text{'}\lambda$. Otherwise,
portfolio $\tau$ is first-order stochastic dominant efficient.
\end{defn}

Second-order stochastic dominance can be similarly defined as first-order:
\begin{defn}
\label{def:SSD-Portfolio-Dominance}Portfolio $\lambda\in\varLambda$
dominates portfolio $\tau\in\varLambda$ by the second-order stochastic
dominance ($\varrho'\lambda\leq_{SD2}\varrho'\tau$) if and only if
\[
F_{\varrho'\lambda}^{(2)}\left(y\right)\leq_{SD2}F_{\varrho'\tau}^{(2)}\left(y\right),\,\,\,\,\,\forall y\in\mathbb{R}
\]
with strict inequality for at least one $y\in\mathbb{R}$ and with
$F_{\varrho'\lambda}^{(2)}\left(y\right)$ denoting the twice cumulative
probability distribution of returns of portfolio $\lambda$. Necessary
and sufficient conditions for the second-order stochastic dominance
($\varrho'\lambda\leq_{SD1}\varrho'\tau$) if:
\end{defn}

\begin{itemize}
\item $Eu\left(\varrho'\lambda\right)\geq Eu\left(\varrho'\tau\right)$
for all expected concave utility functions and strict inequality holds
for at least some concave utility function
\item Non-satiable and risk-averse decision maker prefers portfolio $\tau$
to portfolio $\lambda$ and at least one prefers $\lambda$ to $\tau$
\item $F_{\varrho'\lambda}^{-2}\left(y\right)\leq F_{\varrho'\tau}^{-2}\left(y\right)$
for all $y\in\left[0,1\right]$ with strict inequality for at least
one $y\in\left[0,1\right]$, where $F_{\varrho'\lambda}^{-2}\left(y\right)$
is a cumulated quantile function
\item $CVaR_{\alpha}\left(-\varrho'\lambda\right)\leq CVaR_{\lambda}\left(-\varrho'\tau\right)$
for all $\alpha\in\left[0,1\right]$ with strict inequality for at
least one $\alpha\in\left[0,1\right]$, where
\begin{align*}
CVaR_{\alpha}\left(-r'\lambda\right)=\underset{v\in\mathbb{R},z_{t}\in\mathbb{R}^{+}}{\text{min}} & v+\frac{1}{1-\alpha}\sum_{t=1}^{S}p_{t}z_{t}\\
\text{such that} & z_{t}\geq-x^{t}\lambda-v,\,\,\,\,t=1,2,\ldots,S
\end{align*}
\end{itemize}
\begin{defn}
\label{def:SSD-efficient} A given portfolio $\tau\in\varLambda$
is second-order stochastic dominant (Definition \ref{def:FSD-Portfolio-Dominance})
inefficient if there exists portfolio $\lambda\in\varLambda$ such
that $\varrho\text{'}\lambda\leq_{SD2}\varrho\text{'}\lambda$. Otherwise,
portfolio $\tau$ is second-order stochastic dominant efficient.
\end{defn}

The previous two definitions, first-order and second-order, can be
generalised to the $k$-order:
\begin{defn}
\label{def:sdk-Portfolio--dominance}Portfolio $\lambda$ dominates
portfolio $\tau$ with respect to the $k$-order stochastic dominance
($\lambda\leq_{SDk}\tau$) if $Eu\left(\varrho'\lambda\right)\geq Eu\left(\varrho'\tau\right)$
for all utility functions $u\in U_{n}$ with strict inequality for
at least one such utility function, with $U_{N}$ being the set of
$N$ times differentiable utility functions such that: $\left(-1\right)^{i}u^{\left(i\right)}\leq0$
for all $i=1,2,\ldots,N$.
\end{defn}

~
\begin{defn}
\label{def:SDk-efficient} A given portfolio $\tau$ is SD$k$-efficient
($k\geq2$) if there exists at least one utility function $u\in U_{N}$
such that $Eu\left(\varrho'\tau\right)-Eu\left(\varrho'\lambda\right)\geq0$
for all $\lambda\in\varLambda$ with strict inequality for at least
one $\lambda\in\varLambda$.
\end{defn}

And the previous one period stochastic dominance can be generalised
to the multi-period setting:
\begin{thm}
\label{thm:multiPeriodFSD}( \cite{multiperiodSD} ). Let $F^{n}\left(x\right)$
and $G^{n}\left(x\right)$ be the cumulative distributions of two
$n$-period risks where $n$ is the number of periods and $x$ is
the product of the returns corresponding to each period ($x=x_{1},x_{2},\ldots,x_{n}$).
Then, a sufficient condition for $F^{n}$ dominance over $G^{n}$
by first-order stochastic dominance for every non-decreasing utility
function is that such dominance exists in each period, namely:
\[
F_{i}\left(x_{i}\right)\leq G_{i}\left(x_{i}\right),\,\,\,\,\forall i,\left(i=1,2,\ldots,n\right)
\]
 and there is at least one strict inequality, namely:
\[
F_{i}\left(x_{i_{0}}\right)<G_{i}\left(x_{i_{0}}\right)
\]
for some $x_{i_{0}}$.
\end{thm}

~
\begin{thm}
\label{thm:multiPeriodSSD} ( \cite{multiperiodSD} ). A sufficient
condition for $F^{n}$ dominance over $G^{n}$ by second-order stochastic
dominance for all non-decreasing concave utility functions is that
such dominance exists in each period, namely:
\[
\int_{0}^{x_{i}}\left[G_{i}\left(t_{i}\right)-F_{i}\left(t_{i}\right)\right]dt_{i}\geq0,\,\,\,\,\forall i,\left(i=1,2,\ldots,n\right)
\]
and there is at least one strict inequality.
\end{thm}

Finally, note that the concept of stochastic dominance also extends
to strategy-proof allocation rules and game strategies:~
\begin{defn}
\label{def:(SD-strategyproof)}. A strategy $s$ is \textit{stochastic
dominance strategy-proof} if, for all investors $i\in I$, all security
ranking profiles $\left(R_{i},R_{-i}\right)\in R^{I}$, and all misreports
$R_{i}^{'}\in R$, investor $i$'s assignment $x_{i,j}\in s_{i}\left(R_{i},R_{-i}\right)$
\textit{stochastically dominates} $y_{i,j}\in s_{i}\left(R_{i}^{'},R_{-i}\right)$
at $R_{i}$ (i.e., independent of the other investors' ranking reports),
that is,
\[
\sum_{s_{i}\left(R_{i},R_{-i}\right)}x_{i,j}\geq\sum_{s_{i}\left(R_{i}^{'},R_{-i}\right)}y_{i,j}
\]
Alternatively, \textit{stochastic dominance strategy-proof} can also
be defined in terms of expected utility if, for all utility functions
$u_{i}\in U_{R_{i}}$, we have that
\[
Eu_{s_{i}\left(R_{i},R_{-i}\right)}\left(u_{i}\right)\geq Eu_{s_{i}\left(R_{i}^{'},R_{-i}\right)}\left(u_{i}\right)
\]
All the previous definitions naturally lead to the following theorem
regarding the stochastically dominant crypto-currency:
\end{defn}

\begin{thm}
\label{thm:efficient-portfolio-stochastically-dominant}The efficient
portfolio is to go long on the stochastically dominant crypto-currency:
thus, the stochastically dominant strategy-proof allocation rule for
any investor is to hold this efficient portfolio with the stochastically
dominant crypto-currency. Furthermore, a higher return can be expected
from the stochastically-dominant crypto-currency.
\end{thm}

\begin{proof}
Given a stochastic ranking of policy rules \ref{eq:sdk-optimal-policy}
of a stochastic ordering (Definition \ref{def:(SDk-ordering)}) of
monetary policy rules \ref{eq:rankingPolicyRules} that generates
a stochastically dominant crypto-currency by first-order \ref{prop:FSD-dominant}
or second-order \ref{prop:SSD-dominant} dominance: then, the SD$k$-efficient
portfolio \ref{def:SDk-efficient} (also, first-order \ref{def:FSD-efficient}
or second-order \ref{def:SSD-efficient} efficient) that dominates
with respect to the $k$-order stochastic dominance \ref{def:sdk-Portfolio--dominance}
(also, first-order \ref{def:FSD-Portfolio-Dominance} or second-order
\ref{def:SSD-Portfolio-Dominance} dominant) is the portfolio containing
the stochastically dominant crypto-currency as, by definition, this
the one crypto-currency with the SD$k$-optimal policy rule \ref{eq:sdk-optimal-policy}
that SD$k$ dominates (Definition \ref{def:(SDk-ordering)}) all the
other policy rules. The proof extends trivially to the multi-period
case by Theorems \ref{thm:multiPeriodFSD} and \ref{thm:multiPeriodSSD}. 

Moreover, for all investors $i\in I$, the stochastically dominant
strategy-proof allocation rule \ref{def:(SD-strategyproof)} is to
hold the SD$k$-efficient portfolio \ref{def:SDk-efficient} with
the crypto-currency with the SD$k$-optimal policy rule \ref{eq:sdk-optimal-policy}.

Furthermore, a higher return can be expected from the stochastically-dominant
crypto-currency by the iterated deletion of strictly dominated strategies
when extending Theorem \ref{thm:dominantHigherPrice} to the market
portfolio setting.
\end{proof}
Finally, the efficient and strategy-proof portfolio of Theorem \ref{thm:efficient-portfolio-stochastically-dominant}
induces an investment-efficient Nash equilibrium:
\begin{defn}
A mechanism $M$ induces efficient investment within $\epsilon$ by
investor $i\in I$ if, for all valuation functions $\mathbb{v}^{I\smallsetminus\left\{ i\right\} }\in\mathbb{V}^{I\smallsetminus\left\{ i\right\} }$,
if
\[
\hat{\mathbb{v}}^{i}\in\underset{\tilde{\mathbb{v}}{}^{i}\in\mathbb{V}^{i}}{\text{arg max}}\left\{ E_{\left(\mathbb{v}^{i},\mathbb{v}^{I\smallsetminus\left\{ i\right\} }\right)}\left[u^{i}\left(M\left(v^{i},v^{I\smallsetminus\left\{ i\right\} }\right);v^{i}\right)\right]-c^{i}\left(\tilde{\mathbb{v}}^{i}\right)\right\} 
\]
then we have
\begin{align*}
\left(E_{\left(\hat{\mathbb{v}}^{i},\mathbb{v}^{I\smallsetminus\left\{ i\right\} }\right)}\left[V\left(M\left(v^{i},v^{I\smallsetminus\left\{ i\right\} }\right);\left(v^{i},v^{I\smallsetminus\left\{ i\right\} }\right)\right)\right]\right)-c^{i}\left(\hat{\mathbb{v}}^{i}\right)+\epsilon\\
\geq\underset{\mathbb{v}^{i}\in\mathbb{V}^{i}}{\text{sup}}\left\{ \left(E_{\left(\tilde{\mathbb{v}}^{i},\mathbb{v}^{I\smallsetminus\left\{ i\right\} }\right)}\left[V\left(M\left(v^{i},v^{I\smallsetminus\left\{ i\right\} }\right);\left(v^{i},v^{I\smallsetminus\left\{ i\right\} }\right)\right)\right]\right)-c^{i}\left(\tilde{\mathbb{v}}^{i}\right)\right\} 
\end{align*}
for all cost functions $c^{i}$. In other words, a mechanism induces
efficient investment by $i$ within $\epsilon$ if, assuming agents
report truthfully, every expected utility-maximising investment choice
by $i$ maximises expected social welfare within $\epsilon$.
\end{defn}

\begin{thm}
For any stochastic uncertainties $\epsilon\geq0$ and $\eta\geq0$,
if the portfolio is approximately strategy-proof \ref{def:(SD-strategyproof)}
within $\epsilon$ for investor $i$ and approximately efficient within
$\eta$ (i.e., first-order \ref{def:FSD-efficient}, second-order
\ref{def:SSD-efficient} or SD$k$-efficient \ref{def:SDk-efficient}
), then it induces an approximately efficient investment within $\left(\epsilon+\eta\right)\cdot\left(\text{\#Crypto-currencies}\right)$
to $i$, independent of the other investors' investments. Furthermore,
the stochastically dominant crypto-currency induces a Nash equilibrium
over the crypto-currency market that maximises ex-ante social welfare.
\end{thm}

\begin{proof}
Follows trivially from Theorem 5 and Corollary 2 of \cite{SPequivalenceEfficiency}.
It also holds in expectations for any given investment choice profile
of the other agents using Theorem 7 from \cite{SPequivalenceEfficiency}.
\end{proof}
\begin{rem}
\textbf{On the immovable commitment of crypto-currency monetary policy
rules and the lack of discretion}: most crypto-currencies follow the
example of Bitcoin, where the monetary policy was fixed since its
launch and it was pre-announced that it will never change. This is
in stark contrast with the monetary policy of fiat currencies, where
discretion is preferred in case of a financial crisis. In other words,
the widely accepted monetary policy stance of crypto-currencies to
fix their monetary policies not only leaves them vulnerable to a financial
crisis, but also turns them into dominated crypto-currencies by stochastically
dominant crypto-currencies.
\end{rem}

\subsubsection{\label{subsec:Omega-ratio}Omega ratio of Stochastically Dominant
Crypto-currencies}

The Omega ratio\cite{omegaRatio} is a risk-return performance measure
of an asset, portfolio, or strategy which takes into account all the
higher moment information in the returns distribution and also incorporates
sensitivity to return levels, unlike the Sharpe ratio. It is defined
as the probability-weighted ratio of gains versus losses for some
threshold return target $\theta$,
\[
\Omega\left(\theta\right)=\frac{\int_{\theta}^{\infty}\left[1-F\left(r\right)\right]dr}{\int_{-\infty}^{\theta}F\left(r\right)dr}=\frac{w^{T}E\left(r\right)-\theta}{E\left[\left(\theta-w^{T}r\right)_{+}\right]}+1
\]
Note that first-order stochastic dominance \ref{def:FSD-Portfolio-Dominance}
implies Omega ratio dominance:
\begin{thm}
(Theorem 2, \cite{stochasticDominanceOmegaRatio}). For any two returns
$X$ and $Y$ with means $\mu_{X}$ and $\mu_{Y}$ and Omega ratios
$\Omega_{X}\left(\eta\right)$ and $\Omega_{Y}\left(\eta\right)$,
respectively, if $X\leq_{SD1}Y$, then $\Omega_{X}\left(\eta\right)\geq\Omega_{Y}\left(\eta\right)$
for any $\eta\in R$.
\end{thm}

\begin{cor}
The efficient portfolio long on the stochastically dominant crypto-currency
of Theorem \ref{thm:efficient-portfolio-stochastically-dominant}
has a higher Omega ratio for any return threshold.
\end{cor}

\subsubsection{\label{subsec:Arbitrage-Opportunity}Arbitraging with Stochastic
Dominance}

If there exists a First-order Stochastic Dominance between two assets
\ref{prop:FSD-dominant}, under certain conditions, arbitrage opportunities
will also exist: thus investors will increase not only their expected
utilities, but also their wealth if they shift their holdings to the
dominant asset from the dominated one (i.e., a risk-free investment
opportunity with positive returns).
\begin{thm}
(Arbitrage versus Stochastic Dominance - \cite{FSDarbitrage}). Given
a complete market $M$, there exists an arbitrage opportunity if and
only if there exists assets $x$ and $y\in M$ such that:
\end{thm}

\begin{itemize}
\item $x\leq_{SD1}y$
\item $\left[P_{i}\left(y\leq\alpha\right)-P_{i}\left(x\leq\alpha\right)\right]\leq0$
for all $\alpha\in\mathbb{R}$ and for some investor $i\in I$, where
$P_{i}\left(\right)$is the $i$th investor's subjective probability
belief over the finite number of states of nature
\end{itemize}
In other words, arbitrage implies First-order Stochastic Dominance
but the inverse is not necessarily true: it's only true when the cumulative
distribution functions of the assets are perfectly correlated or the
risky asset is a monotone function of the asset even in the absence
of perfect correlation. Note that crypto-currency markets are unusually
highly correlated compared to other asset markets.

In practice, empirical studies may statistically detect First-order
Stochastic Dominance, but arbitrage opportunities may not exist: nonetheless,
investors can increase their expected utilities, as well as their
expected wealth, if they shift their holdings to the dominant asset
from the dominated one \cite{FSDNonarbitrage}.

\section{\label{sec:Model-Policies}Model and Policies}

DSGE models constitute the modern workhorse of monetary policy analysis,
with a recent survey finding 84 models used by 58 institutions \cite{dsgeSurvey}.
In this section, a Dynamic Stochastic General Equilibrium (DSGE) parsimonious
model is introduced to an economy featuring a Central-Bank Digital
Currency (CBDC) and a crypto-currency, calibrated and estimated for
the United States.

The model is further simplified by opting for a closed-economy instead
of a small open-economy, justified by previous results showing that
optimal policies under parameter uncertainty lack exchange rate responses
\cite{monetaryUncertaintySOE} and that welfare loss functions for
small open economies do not include foreign variables when the calibration
is imposed \cite{monetaryExchangeRateVolatilitySOE}.

\subsection{\label{subsec:Monetary-Policy-Rules}Monetary Policy Rules}

The following monetary policy rules are implemented in this model.

\subsubsection{Central Bank}

Taylor's rule for the monetary policy:
\begin{equation}
\begin{array}{c}
{{i}_{t}}=\left({\rho_{2}}-{\rho_{1}}\right)\,(\bar{{i}})+{\rho_{1}}\,{{i}_{t-1}}+{\rho_{3}}\,\left(1-\frac{{{M}_{t-1}}}{{{Y}_{t-1}}}\right)\\
+\left({\rho_{2}}-{\rho_{1}}\right)\,\left({\phi_{\pi}}\,\left({{\pi}_{t}}-{\bar{\pi}}\right)+{\phi_{y}}\,\left(log\left({{Y}_{t}}\right)-log\left({{Y}_{t-1}}\right)\right)\right)+{{\epsilon_{i}}_{t}}
\end{array}
\end{equation}
where $\rho_{i}$ are smoothing parameters and $\phi_{\pi}$ is the
inflation feedback coefficient.

\subsubsection{Bitcoin's Monetary Policy\label{subsec:Bitcoin's-Monetary-Policy}}

Although Bitcoin's monetary policy is not described in its paper,
its implementation appears in the source code \cite{bitcoinMonetaryPolicy}:
the initial reward of 50 BTC is halved every 210,000 blocks (4 years),
and each block is mined approximately every 10 minutes. The supply
formula in block time is given by,
\begin{align}
B\left(t\right) & =\sum_{t=1}^{min\left(t,T\right)}\frac{50}{2^{H\left(t\right)}}\label{eq:bitcoinSupply}\\
H\left(t\right) & =\left\lfloor \frac{t}{210000}\right\rfloor 
\end{align}
where $t$ is the block height, $H\left(t\right)$ is the number of
reward halvings up to block $t$, and $T=33x210000$. Bitcoin supply
is limited beyond block $T$, given by
\begin{equation}
B_{total}=\sum_{i=0}^{32}\frac{50}{2^{i}}\times210000\approx2.1\times10^{7}\label{eq:bitcoinMaxIssuance}
\end{equation}
Equation \ref{eq:bitcoinSupply} can be fitted as an exponential curve,
given by:
\begin{equation}
S_{t}^{BTC}=2.1\times10^{7}\times\left(1-\alpha^{t}\right)\label{eq:bitcoinSupplyPeriod}
\end{equation}
where $S_{t}$ is the supply in period $t$, and $\alpha$ is the
growth rate with $\alpha\approx0.825$ for yearly periods and $\alpha\approx0.953$
for quarterly periods. The previous exponential curve can be rewritten
in recursive form as:
\begin{align}
f_{0} & =2.1\times10^{7}\\
f_{t} & =0.825\times f_{t-1}\\
S_{t}^{BTC} & =2.1\times10^{7}-f_{t}
\end{align}
or equivalently:
\begin{align}
S_{t+1}^{BTC} & =S_{t}^{BTC}+(1-\alpha)\left(2.1\times10^{7}-S_{t}^{BTC}\right)\\
 & =\alpha\cdot S_{t}^{BTC}+\left(1-\alpha\right)\cdot2.1\times10^{7}
\end{align}

The following figure displays the evolution of bitcoin supply assuming
exact 10-minute confirmation times.

\begin{figure}[H]
\begin{centering}
\includegraphics[scale=0.2]{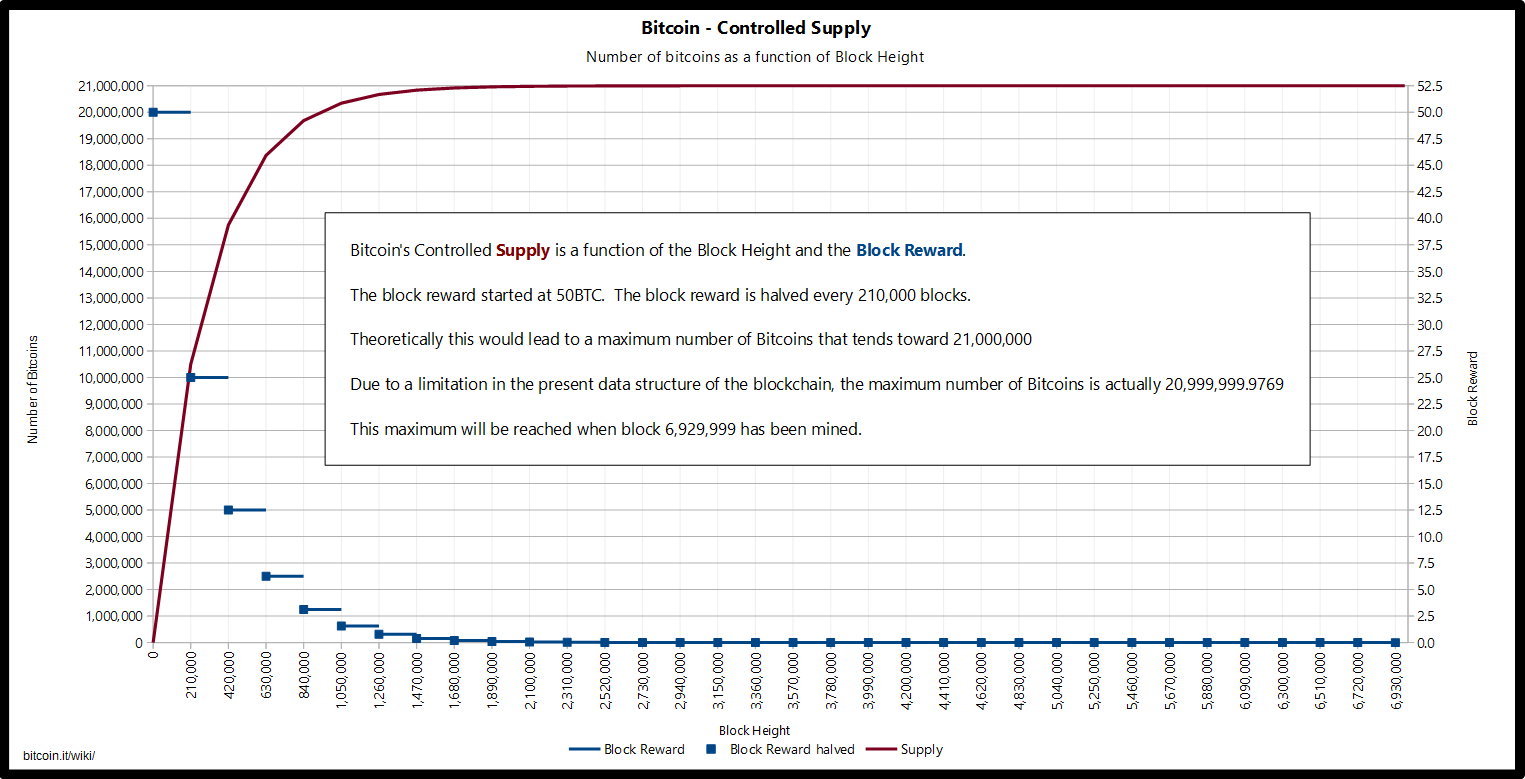}
\par\end{centering}
\caption{\label{fig:Bitcoin's-controlled-supply}Bitcoin's controlled supply\cite{bitcoinControlledSupply}}
\end{figure}

As previously pointed out in Figure \ref{fig:Relative-supply-of-cryptocurrencies}
comparing the relative supply of crypto-currencies, most crypto-currencies
follow similar supply curves but use different parameters: therefore,
without loss of generality we will only consider Bitcoin in this paper
in representation of all the other crypto-currencies.

Note that Bitcoin's monetary policy is independent of any observable
variable (e.g., inflation, output, ....) and Satoshi Nakamoto pre-committed
not to ever modify it: in macro-economics, this monetary policy can
be interpreted as a deflationary version of Friedman's $k$-percent
rule\cite{USmonetaryHistory} (i.e., constant money growth). Following
Poole's classical Keynesian analysis \cite{poole70} in a stochastic
IS-LM model, monetary policies targeting only the money stock allow
money demand shocks to contribute to macroeconomic volatility: indeed,
recent analysis in modern New Keynesian models \cite{fedSince1980,pooleNK,galiBook}
demonstrate that constant money growth rules lead to excess volatility
in both output and inflation when the economy faces money demand shocks,
or other disturbances that require output and inflation to adjust.
This situation is further aggravated by an inelastic supply curve
in both the short and the long term: as the following comparative
chart shows \ref{fig:Elastic-v-Inelastic}, supply inelasticities
imply dramatic price changes with even minor changes in demand, thus
contributing to Bitcoin's volatility.

\begin{figure}[H]
\begin{centering}
\includegraphics[scale=0.45]{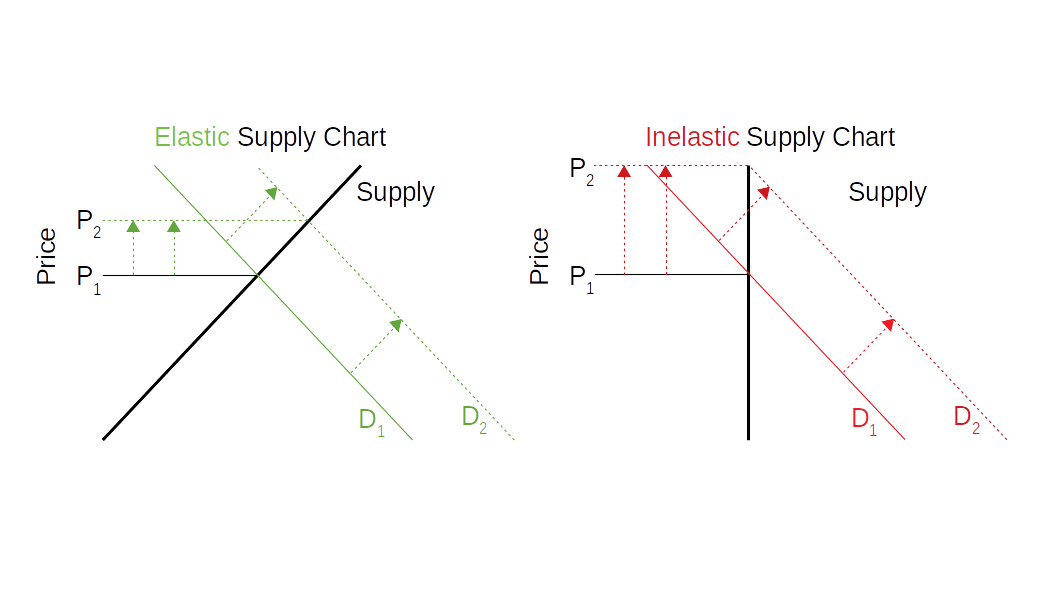}
\par\end{centering}
\caption{\label{fig:Elastic-v-Inelastic}Elastic v. Inelastic Supply Charts}

\end{figure}

However, money growth rules perform much better when they are able
to adjust to movements in the output gap and inflation as exemplified
by the two following rules \ref{eq:maccallumRule} and \ref{eq:moneyRulesReconsidered}:
advantageously, these money growth rules are able to stabilise inflation
by pre-committing to an average rate of money growth and focusing
directly on stabilising the output gap over shorter time horizons,
instead of the aggressive responses to inflation needed by interest
rate rules (i.e., Taylor's rule).

\subsubsection{\label{subsec:Ethereum-2.0's-Monetary-Policy}Ethereum 2.0's Monetary
Policy}

Other crypto-currencies feature a much more complex monetary policy
than Bitcoin's monetary policy \ref{subsec:Bitcoin's-Monetary-Policy},
although in essence they all suffer from the same shortcoming: they
fail to react to changes in inflation, output gap, or any other macro-economic
aggregate (unlike the monetary policy presented in this paper).

For the particular case of Ethereum 2.0 after transitioning to Proof-of-Stake
(a.k.a., ``the Merge''), the monetary policy will be described by
the following features:
\begin{itemize}
\item almost deflationary by default: issuance reduced from 2 Ether/block
to a variable number depending on the total amount of Ether at stake
(currently around 13.3MM ETH), which will be around 600K ETH/year,
implying a 90\% reduction
\item deflationary burning of transaction fees (EIP-1559)
\item double use as store of value and gas for smart contracts
\end{itemize}
Accounting in a monetary policy rule for all the previous features
will only make it more deflationary, thus less reactive to changes
in the macro-economic environment (i.e., a narrower path of policy
responses) and therefore much more stochastically dominated even than
Bitcoin's monetary policy \ref{subsec:Bitcoin's-Monetary-Policy}.

\subsubsection{\label{subsec:McCallum's-Policy-Rule}McCallum's Policy Rule}

A classical monetarist policy, McCallum's rule \cite{mccallumPolicyRule}
is specified by:
\begin{equation}
\Delta b_{t}=\Delta x^{*}-\frac{\left(x_{t-1}-b_{t-1}-x_{t-17}+b_{t-17}\right)}{16}+\lambda\left(x_{t-1}^{*}-x_{t-1}\right)\label{eq:maccallumRule}
\end{equation}
where the previous variables are defined as:
\begin{itemize}
\item $b_{t}$ is the logarithm of the adjusted monetary base
\item $x_{t}$ is the logarithm of the adjusted nominal GDP
\item $x_{t}^{*}$ is the target value of $x_{t}$ for quarter $t$ (growing
smoothly at the rate $\Delta x^{*}$).
\end{itemize}
The second term provides a velocity growth adjustment intended to
reflect long-lasting institutional changes, while the third term features
feedback adjustment in $\Delta b_{t}$ in response to cyclical departures
of $x_{t}$ from the target path $x_{t}^{*}$, with $\lambda\geq0$
chosen to balance the speed of eliminating $x_{t}^{*}-x_{t}$ gaps
against the danger of instrument instability.

\subsubsection{\label{subsec:Reconsideration-of-Money-Growth}A Reconsideration
of Money Growth Rules}

If the Federal Reserve would have used a money rule targeting money
growth instead of the interest rate during the 2007-2009 recession,
the US economy would have recovered more quickly, and during the 2009-2015
period of zero nominal interest rates, it would have stabilised output
and inflation with comparable performance \cite{moneyRulesReconsidered}.
While the recent consensus was that policy rules using \textit{constant}
rates of money growth would have performed poorly in comparison to
Taylor rules, recent work \cite{moneyRulesReconsidered} shows that
money growth rules augmented to adjust to movements in the output
gap and inflation in a manner similar to the Taylor rule will perform
significantly better, on par with more conventional Taylor rules for
the interest rate. Thus, the reconsidered money growth rule is given
by
\begin{equation}
\ln\left(\mu_{t}/\mu\right)=\rho_{mm}\ln\left(\mu_{t-1}/\mu\right)+\rho_{m\pi}\ln\left(\pi_{t}/\pi\right)+\rho_{mx}\ln\left(x_{t}/x\right)\label{eq:moneyRulesReconsidered}
\end{equation}
where the previous variables are defined as:
\begin{itemize}
\item $\mu_{t}=M_{t}/M_{t-1}$ denotes the growth rate of nominal money
\item $\mu$ denotes the steady-state rate of money growth
\item $\pi$ denotes the steady-state rate of inflation
\item $x$ denotes the steady-state values of the output gap
\end{itemize}
Depending on the values of the parameters, the following cases can
be considered:
\begin{itemize}
\item $\rho_{m}=\rho_{m\pi}=\rho_{mx}=0$ is the \textit{constant} money
growth rule as advocated by Friedman \cite{USmonetaryHistory}
\item $\rho_{m\pi}<0$ and $\rho_{mx}<0$ allow to stabilise inflation and
the output gap in response to shocks
\item $\rho_{m\pi}<0$, $\rho_{mx}<0$ and $\rho_{mm}>0$ prescribe a gradual
response of money growth to movements in inflation and the output
gap, much like the Taylor rule with interest rate smoothing
\end{itemize}

\subsection{\label{subsec:Ranking-Policy-Rules}Ranking of Policy Rules}

\includegraphics[scale=0.22]{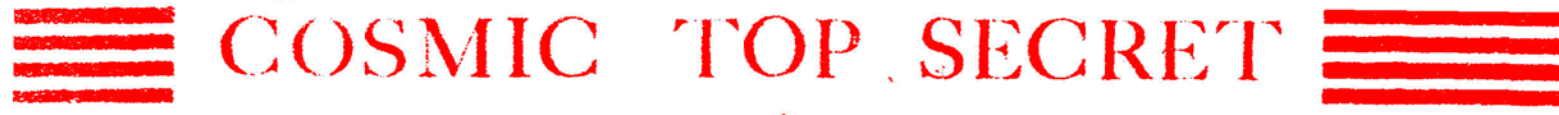}

\section{Implementation Details\label{sec:Implementation-Details}}

The calculation of the stochastically-dominant optimal monetary policy
is implemented using Dynare \cite{dynare5} with an additional 225.000
MATLAB/Octave LOCs.

\subsection{\label{subsec:Global-Implementation}Global Implementation}

Different countries feature different macro-economic indicators (inflation,
interest rate, output, GDP growth, exchange rates...), thus it is
very important for the consensus protocol to be aware of the different
nationalities of its participants (nodes and/or users): Pravuil \cite{pravuil}
is specifically designed for an international setting as it integrates
national identity cards and biometric passports in layer 1, making
it ideal to implement different monetary policies in different countries.

Furthermore, the combination of Zero-Knowledge Proof of Identity\cite{cryptoeprint:2019/546}
with the Zero-Knowledge Stochastically Dominant crypto-currency induces
the following pincer manoeuvre:

\begin{figure}[H]
\begin{centering}
\includegraphics[scale=0.4]{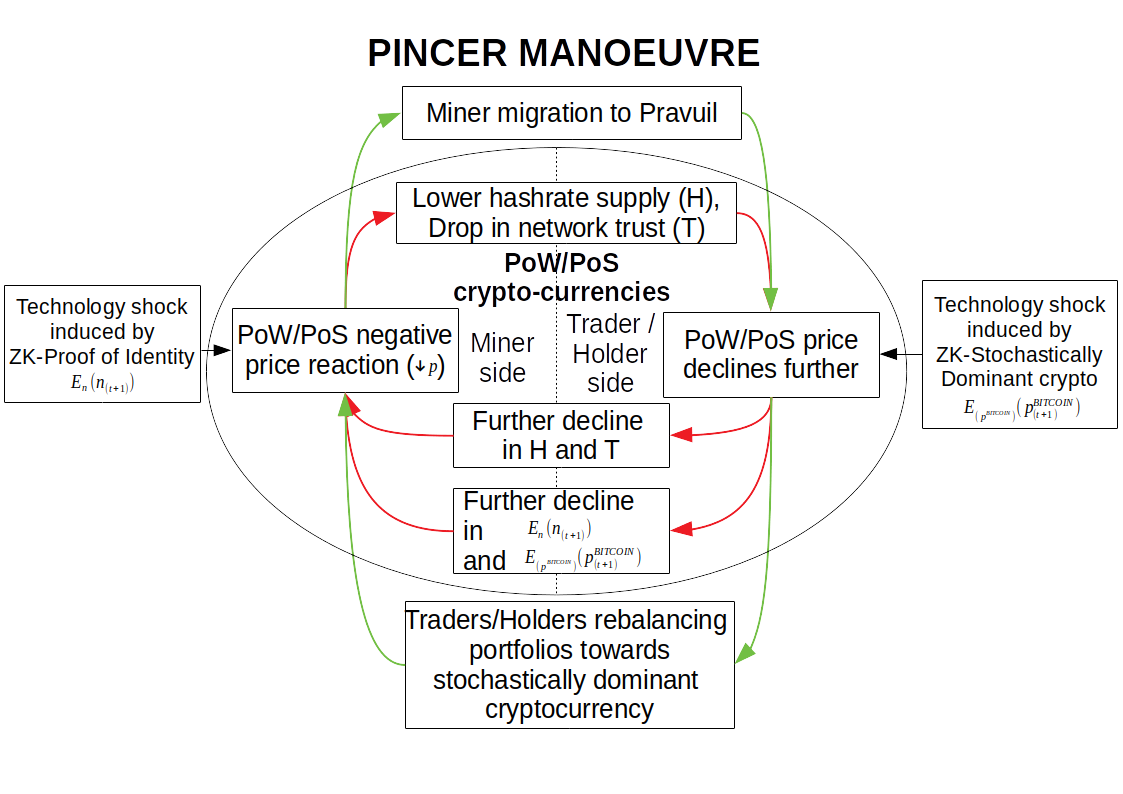}
\par\end{centering}
\caption{\label{fig:pincer-Maneuver}Pincer manoeuvre inducing a downward spiral
on PoW/PoS crypto-currencies (in \textcolor{red}{red}) and a virtuous
cycle for the Zero-Knowledge Stochastically Dominant crypto-currency
(in \textcolor{green}{green})}

\end{figure}

\subsection{\label{subsec:Zero-Knowledge-Monetary-Policy}Zero-Knowledge Monetary
Policy}

To understand the reason behind the lack of advanced monetary policies
in crypto-currencies as the ones described in this paper in subsections
\ref{subsec:Monetary-Policy-Rules}, one has to look back to a reply
by Satoshi Nakamoto \cite{satoshiNoCentralBank} on its original post
announcing the first implementation of Bitcoin:
\begin{quotation}
\textit{Indeed there is nobody to act as central bank or federal reserve
to adjust the money supply as the population of users grows. That
would have required a trusted party to determine the value, because
I don't know a way for software to know the real world value of things.
If there was some clever way, or if we wanted to trust someone to
actively manage the money supply to peg it to something, the rules
could have been programmed for that.}
\end{quotation}
Fortunately, the author of this paper is more knowledgeable: this
subsection describes a zero-knowledge protocol to securely compute
monetary policies using authenticated economic series and commit their
resulting zero-knowledge proofs on the blockchain.

\begin{figure}[H]
\centering{}\includegraphics[scale=0.45]{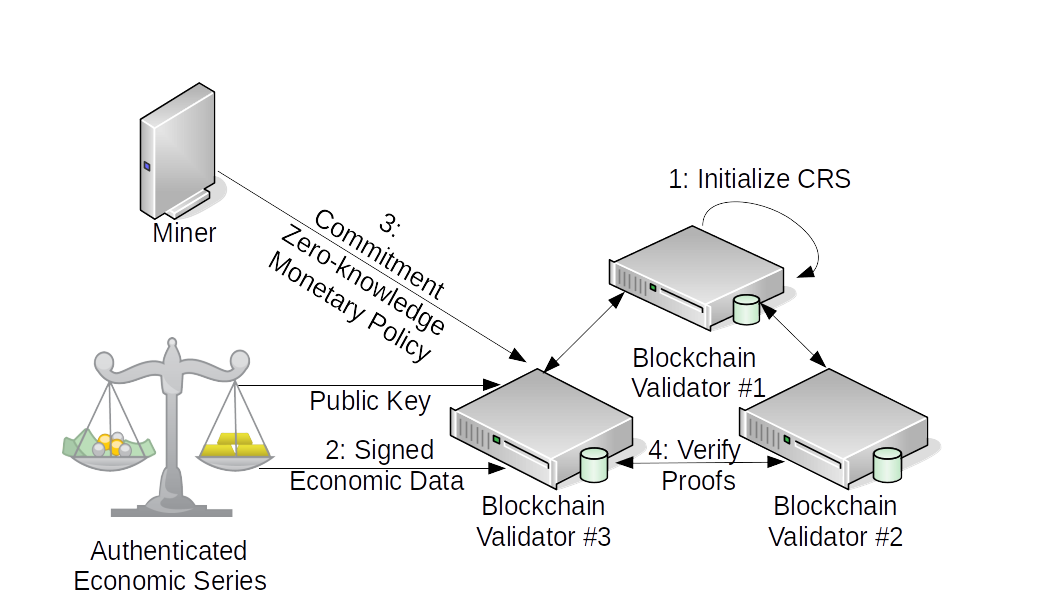}\caption{\label{fig:Committing-ZK-monetary-policies}Committing zero-knowledge
monetary policies on a blockchain}
\end{figure}
As pictured in the Figure \ref{fig:Committing-ZK-monetary-policies}
above, there are 3 parties to the protocol:
\begin{itemize}
\item \textbf{Miner}: commits transactions to the blockchain and gets rewarded
according to a monetary policy rule using data from providers of authenticated
economic series, and optionally its own private data.
\item \textbf{Providers of Authenticated Economic Series}: take economic
series from public providers (e.g., FRED, db.nomics, ...) and authenticate
their data on the blockchain by signing with their private keys $sk_{econ}$,
so it can later be verified by everyone with the public key $pk_{econ}$.
\item \textbf{Blockchain Validators}: blockchain nodes that verify transactions,
blocks, and proofs. As previously discussed \ref{subsec:Global-Implementation},
they should be running the Pravuil \cite{pravuil} consensus protocol.
\end{itemize}

\subsubsection{Security Model}

The security model is defined with an ideal functionality $\mathcal{F}{}_{zkMonetaryPolicy}$
that rigorously sets the security requirements of the zero-knowledge
protocol:
\begin{itemize}
\item \textbf{Initialisation}: the blockchain is initialised with public
input data $p$ and a computation circuit $C$ 
\item \textbf{AuthenticateEconomicData}: the provider of Authenticated Economic
Series sends a data authentication request to obtain the digital signature
$s_{econ}$ over $\left(\text{data}_{public}\right)$.
\item \textbf{zk-CommitMonetaryPolicy}: miners request with authenticated
data containing $\text{input}_{public},input_{private}$, the hash
$h$ of inputs, and the $\text{output}_{miner}$.
\end{itemize}
\noindent\fbox{\begin{minipage}[t]{1\columnwidth - 2\fboxsep - 2\fboxrule}%
\begin{center}
\textbf{\label{Ideal-Functionality zkMonetaryPolicy}Ideal Functionality}
$\mathcal{F}{}_{zkMonetaryPolicy}$
\par\end{center}
\begin{flushleft}
$\mathcal{F}_{zkMonetaryPolicy}$ interacts with the adversary $\mathcal{A}$,
the miner, the providers of authenticated economic series, the ideal
functionality $\mathcal{F}{}_{sig}$ and the ideal blockchain ledger
functionality $\mathcal{L}$ with the following queries:
\par\end{flushleft}
\begin{itemize}
\item \textbf{Initialisation}: upon receiving $\left(init,C,p\right)$ on
initialisation:
\begin{itemize}
\item store the circuit $C$ and the public input data $p$
\item send $\left(init,C,p\right)$ to $\mathcal{A}$
\end{itemize}
\item \textbf{AuthenticateEconomicData}: upon receiving $\left(authenticate,\text{data}_{public}\right)$
from a provider of authenticated economic series:
\begin{itemize}
\item send $\left(sign,provider,\text{data}_{public}\right)$ to $\mathcal{F}{}_{sig}$
and receives signature $s_{econ}$
\item send $\left(sign,provider,\text{data}_{public}\right)$ to $\mathcal{A}$
\end{itemize}
\item \textbf{Validate}: upon receiving $\left(validate,\text{output}_{miner},\text{input}_{public},\text{input}_{private},h,s_{econ}\right)$
from a miner:
\begin{itemize}
\item send $\left(verify,provider,h,s_{econ}\right)$ to $\mathcal{F}{}_{sig}$
and check that it's correct
\item check that $\left(p,\text{output}_{miner},\text{input}_{public},\,\text{input}_{private},h\right)$
satisfies the circuit $C$
\item send $\left(validate,\text{output}_{miner},\text{input}_{public},h,s_{econ}\right)$
to $\mathcal{A}$
\end{itemize}
\end{itemize}
\end{minipage}}

The ideal functionality $\mathcal{F}_{zkMonetaryPolicy}$ captures
the following design goals:
\begin{itemize}
\item \textbf{Authenticity}: blockchain validators execute only on resulted
computations from providers of authenticated economic series, rejecting
otherwise.
\item \textbf{Privacy}: the private data of the miner is never exposed to
anyone, and the blockchain validators are executed correctly without
the private data using the zero-knowledge proof.
\end{itemize}

\subsubsection{Protocol Description and Implementation}

Using a zero-knowledge SNARK scheme $\varLambda$, the steps of the
proposed scheme would be as follows:
\begin{itemize}
\item \textbf{Initialisation}: A security parameter $1^{\lambda}$ is picked
in accordance with the security requirements, and a circuit $C$ is
constructed for the computation over the authenticated data. Then,
a trusted generator or a MPC protocol setups the zk-SNARK with $\left(1^{\lambda},C\right)$
to create the Common Reference String for proof generation and verification.\\
Concurrently, the provider of authenticated economic series chooses
a public/private key pair $\left(pk_{econ},sk_{econ}\right)$. Only
then, $\left(CRS,pk_{econ}\right)$ are published on the blockchain
for everyone to check their validity.
\item \textbf{AuthenticateEconomicData}: providers of authenticated economic
series obtain signatures $s_{econ}$ with parameters $\left(sk_{econ},\left(h,\text{data}_{public}\right)\right)$.
\item \textbf{zk-CommitMonetaryPolicy}: miner uses circuit $C$ of the monetary
policy to obtain the result $\text{output}_{miner}$ and a hash $h$
of $\left(\text{input}_{public},\text{input}_{private}\right)$; then,
the miner executes the zk-SNARK for proving with parameters $\left(CRS,p,\text{input}_{public},\text{input}_{private},\text{output}_{miner},h\right)$
obtaining the zero-knowledge proof $\pi$. Then, the miner sends a
transaction to the blockchain validators as follows:
\[
\text{tx}_{sk_{miner}}=\left(validate,\pi,\text{input}_{public},\text{output}_{miner},h\right)
\]
\item \textbf{Validation}: blockchain validators verify the zk-SNARK with
parameters $\left(CRS,pk_{econ},\pi,p,\text{input}_{public},\text{output}_{miner},h\right)$:
only in case it's found valid, then the block from the miner is accepted
with the computed monetary policy.
\end{itemize}
\noindent\fbox{\begin{minipage}[t]{1\columnwidth - 2\fboxsep - 2\fboxrule}%
\begin{center}
\textbf{\label{zkMonetaryPolicy-Protocol}$zk-MonetaryPolicy$ Protocol}
\par\end{center}
\textbf{Miner}:
\begin{itemize}
\item \textbf{zk-CommitMonetaryPolicy}: on input $\left(commit,p,\text{input}_{public},\text{input}_{private},\text{output}_{miner},h\right)$
\begin{itemize}
\item prove with zk-SNARK: 
\[
\pi=Prove\left(CRS,p,\text{input}_{public},\text{input}_{private},\text{output}_{miner},h\right)
\]
\item send $\text{tx}_{sk_{miner}}=\left(validate,\pi,\text{input}_{public},\text{output}_{miner},h\right)$
to the blockchain validator
\end{itemize}
\end{itemize}
\textbf{Providers of Authenticated Economic Series}:
\begin{itemize}
\item \textbf{Initialisation}:
\begin{itemize}
\item $\left(pk_{econ},sk_{econ}\right)=KeyGeneration\left(1^{\lambda}\right)$
\end{itemize}
\item \textbf{Commit Authenticated Economic Data}:
\begin{itemize}
\item compute $h=Hash(data_{public})$ and $s_{econ}=Sign\left(sk_{econ},\left(h,data_{public}\right)\right)$
\item send $\left(h,s_{econ}\right)$ to the blockchain
\end{itemize}
\end{itemize}
\textbf{Blockchain Validators:}
\begin{itemize}
\item \textbf{Initialisation}: upon receiving $\left(init,C,p,CRS,pk_{econ}\right)$
\begin{itemize}
\item Store the public input data $p$ for $C$
\item Store the common reference string CRS and $pk_{econ}$
\end{itemize}
\item \textbf{Validation}: upon receiving $\left(validate,\pi,\text{input}_{public},\text{output}_{miner},h\right)$
\begin{itemize}
\item Check that $h$ is stored on the blockchain
\item Check that zk-SNARK $\left(CRS,pk_{econ},\pi,p,\text{input}_{public},\text{output}_{miner},h\right)$
is valid
\item If valid, proceed to store the transactions, block, and associated
zk-proof $\pi$
\end{itemize}
\end{itemize}
\end{minipage}}

The following theorem formalises the security and privacy of the above
scheme:
\begin{thm}
\label{thm:securityProof}If $\varLambda$ is a simulation-extractable
zk-SNARK with data authentication scheme, then the above scheme is
a privacy-preserving scheme under the universally composable framework.
\end{thm}

\begin{proof}
See \ref{proof:securityProof}.
\end{proof}
\begin{cor}
In the implementation, a simulation-extractable zk-SNARK such as Plonk
must be used \cite{cryptoeprint:2021/511}, even if it has larger
proofs than other more succinct zk-SNARKs.
\end{cor}

An implementation in Go using gnark\cite{gnark-v0.6.4} is available
at \url{https://github.com/Calctopia-OpenSource/cothority/tree/zkmonpolicy}

\section{\label{sec:Conclusion}Conclusion}

The present paper has tackled and successfully solved the problem
of optimal monetary policies specifically tailored for crypto-currencies,
stochastically dominating all the other previous crypto-currencies.
Furthermore, the efficient portfolio is to hold the stochastically
dominant crypto-currency implementing the optimal monetary policy,
a strategy-proof arbitrage featuring a higher Omega ratio with a higher
expected return, inducing a Nash equilibrium over the crypto-currency
market.

\bibliographystyle{alpha}
\bibliography{bib}

\pagebreak{}

\part{Appendix}
\begin{proof}
\textbf{\label{proof:securityProof}(Theorem \ref{thm:securityProof}
)}. The protocol \ref{zkMonetaryPolicy-Protocol} securely realises
the ideal functionality $\mathcal{F}_{zkMonetaryPolicy}$ \ref{Ideal-Functionality zkMonetaryPolicy}:
by using the universal composability framework, we first show an ideal-world
simulator for the dummy adversary $\mathcal{A}$ automatically passing
messages to and from the actual adversary, the environment $\mathcal{E}$;
then, we show the indistinguishability of the ideal and the $\mathcal{F}_{sig}$-Hybrid
worlds.

\textbf{\textit{Ideal-world simulator}}. For conciseness, we only
focus on the simulator $\mathcal{S}$ and not on the blockchain functionality.

- \textbf{Initialisation}: simulator $\mathcal{S}$ obtains $\widehat{CRS}$
and a trapdoor $\tau$ by running a simulated setup algorithm of the
zk-SNARK scheme $\Lambda$. Then, simulator $\mathcal{S}$ keeps $\tau$
and sends $\widehat{CRS}$ to $\mathcal{E}$.

- \textbf{Simulating honest parties} (note that only \textbf{zk-CommitMonetaryPolicy}
needs to be simulated): $\mathcal{E}$ sends $\left(validate,\text{output}_{miner},\text{input}_{public},\text{input}_{private},h\right)$
to an honest miner and simulator $\mathcal{S}$ receives $\left(validate,\text{output}_{miner},\text{input}_{public},h\right)$
from the ideal functionality $\mathcal{F}_{zkMonetaryPolicy}$; then,
simulator $S$ generates an indistinguishable proof $\pi$ using trapdoor
$\tau$ (i.e., without knowing $\text{input}_{private}$). Finally,
$\mathcal{S}$ sends $\left(validate,\pi,\text{input}_{public},\text{output}_{miner},h\right)$
to the blockchain validators.

- \textbf{Simulating corrupted parties}: $\mathcal{E}$ requests to
the simulator $\mathcal{S}$ on behalf of corrupted parties; then
$\mathcal{S}$ processes as follows: $\mathcal{S}$ receives $(validate,$
$\text{output}_{miner},$ $\text{input}_{public},$ $\text{input}_{private},$
$h)$ and extracts $\text{input}_{private}$ from the proof $\pi$
using the trapdoor $\tau$, then sends $\left(validate,\pi,\text{input}_{public},\text{output}_{miner},h\right)$
to $\mathcal{F}{}_{zkMonetaryPolicy}$.

\textbf{\textit{Indistinguishability between the ideal and the $\mathcal{F}_{sig}$-Hybrid
worlds}}: a series of games from the $\mathcal{F}_{sig}$-Hybrid protocol
execution until the ideal world.

- \textbf{$\mathcal{F}_{sig}$-Hybrid model}: a dummy adversary passes
messages for the environment $\mathcal{E}$, the actual adversary.

- \textbf{Hybrid $\mathcal{H}_{1}$}: adds to the $\mathcal{F}_{sig}$-Hybrid
world calls to the simulated setup that generates $\tau$ (kept by
the simulator) and $\widehat{CRS}$, sent to $\mathcal{E}$. $\mathcal{H}_{1}$
replaces the real proofs with the simulated proofs using $\widehat{CRS}$
and $\tau$: due to the computational zero-knowledge property, $\mathcal{H}_{1}$
is computationally indistinguishable from the $\mathcal{F}_{sig}$-Hybrid
world.

- \textbf{Hybrid $\mathcal{H}_{2}$}: adds the simulation of the blockchain
to the $\mathcal{H}_{1}$ world. From the adversary $\mathcal{E}$'s
point of view, $\mathcal{H}_{2}$ is indistinguishable from $\mathcal{H}_{1}$
because the blockchain functionality is public.

- \textbf{Hybrid $\mathcal{H}_{3}$}: adds to the $\mathcal{H}_{2}$
world, the extraction of the private witness from a zero-knowledge
proof $\pi$ by $\mathcal{S}$ if it is a valid proof, otherwise aborts.
$\mathcal{H}_{3}$ is indistinguishable from $\mathcal{H}_{2}$ because
the abort probability is negligible due to the simulation extractability
property of the zk-SNARK.

Finally, the ideal and the $\mathcal{F}_{sig}$-Hybrid world are computationally
indistinguishable because $\mathcal{H}_{3}$ is computationally indistinguishable
for $\mathcal{E}$ from the ideal simulation. Note that any universal-composable
signature scheme can implement $\mathcal{F}_{sig}$ due to the universal
composition theorem.
\end{proof}

\end{document}